\documentclass[onecolumn]{revtex4-2}
\usepackage[
	bookmarks  = false,
	colorlinks = true,
	linkcolor  = blue,
	urlcolor   = blue,
	citecolor  = blue]{hyperref}
\usepackage{amsmath,amssymb,amsthm,braket}
\usepackage{graphicx,bm,bbm}
\usepackage[T1]{fontenc}
\usepackage{microtype}
\usepackage{orcidlink}
\newcommand{\appsection}[1]{\section{\MakeUppercase{#1}}}
\theoremstyle{remark}
\newtheorem{theorem}{Theorem}
\newtheorem{lemma}{Lemma}

\newtheorem{definition}{Definition}

\newtheorem{example}{Example}
\DeclareMathOperator{\Tr}{Tr}
\begin{document}
\title{Factorizability of optimal quantum sequence discrimination \\
under maximum-confidence measurements}
\author{Donghoon Ha\,\orcidlink{0000-0002-1698-9584}}
\affiliation{Department of Applied Mathematics and Institute of Natural Sciences, Kyung Hee University, Yongin 17104, Republic of Korea}
\author{Jeong San Kim\,\orcidlink{0000-0002-0420-8955}}
\email{freddie1@khu.ac.kr}
\affiliation{Department of Applied Mathematics and Institute of Natural Sciences, Kyung Hee University, Yongin 17104, Republic of Korea}
\begin{abstract}
We consider the discrimination of quantum sequences under maximum-confidence measurements and show that the optimal discrimination of a quantum sequence ensemble can always be factorized into that of each individual ensemble.
In other words, the optimal quantum sequence discrimination under maximum-confidence measurements can be achieved just by performing a maximum-confidence discrimination independently at each step of the quantum sequence.
We also show that the maximum confidence of identifying a quantum sequence is to achieve the maximum confidence of identifying each state comprising the quantum sequence.
We further provide a necessary and sufficient condition for the optimal quantum state discrimination under maximum-confidence measurements.
\end{abstract}
\maketitle

\section{Introduction}\label{sec:intro}
Quantum state discrimination is one of the fundamental concepts used in various quantum information processing tasks\cite{chef2000,berg2007,barn20091,bae2015}.
In general, we can always discriminate orthogonal quantum states using an appropriate measurement.
On the other hand, no measurement can perfectly discriminate nonorthogonal quantum states.
For these reasons, various state discrimination strategies have been developed to optimally discriminate nonorthogonal quantum states, such as minimum-error discrimination, unambiguous discrimination and maximum-confidence discrimination\cite{hels1969,ivan1987,diek1988,pere1988,crok2006}.

When a discrimination process involves multiple quantum state ensembles, the situation becomes significantly more intricate, as it requires identifying not just a single quantum state but a sequence of states where each state is independently prepared from a quantum state ensemble.  
Quantum sequence discrimination naturally appears in various quantum information tasks, including multiple-copy state discrimination and quantum change-point detection\cite{higg2011,sent2016,sent2017}. 
Thus, a deeper understanding of quantum sequence discrimination provides valuable insights into the general principles of quantum information processing.

The optimal discrimination of a quantum sequence ensemble can be factorized into that of each individual ensemble. 
In other words, the optimal discrimination of quantum sequences can be realized just by performing a quantum state discrimination independently at each step of the quantum sequence. 
This factorizability always holds in minimum-error discrimination and optimal unambiguous discrimination\cite{gupt20241,gupt20242}. 
Recently, a factorizability condition was established in the optimal discrimination of multi-party quantum sequences under local operations and classical communication\cite{ha2025}.

The factorizability of quantum sequence discrimination characterizes the fundamental limits of quantum information processing tasks such as quantum cryptography, quantum teleportation and quantum data hiding\cite{benn1986,benn1993,divi2002}. 
When a quantum information processing task based on quantum state discrimination is independently repeated multiple times, the factorizability inherent in the underlying discrimination strategy implies that quantum memory or collective measurements offer no advantage.
As there exist quantum information processing tasks based on maximum-confidence discrimination\cite{benn1992,neve2012}, it is a natural question to ask whether factorizability still holds in the optimal quantum sequence discrimination under \emph{maximum-confidence measurements}(MCMs).

Here, we provide an answer to the question by showing that the optimal discrimination of a quantum sequence ensemble under MCMs can always be factorized into that of each individual ensemble. 
In other words, the optimal quantum sequence discrimination under MCMs can be achieved just by performing a maximum-confidence discrimination independently at each step of the quantum sequence. 
We also show that the maximum confidence of identifying a quantum sequence is to achieve the maximum confidence of identifying each state in the quantum sequence. 
Moreover, we establish a necessary and sufficient condition for the optimal quantum state discrimination under MCMs.
In particular, our results naturally include the known factorizability results for optimal unambiguous discrimination as a special case when the confidence of each quantum sequence is equal to one.

This paper is organized as follows.
In Sec.~\ref{ssec:qsd}, we review the definitions and properties of the confidence of measurement in quantum state discrimination.
In Sec.~\ref{ssec:omcd}, we provide a necessary and sufficient condition for the optimal quantum state discrimination under MCMs.
In Sec.~\ref{ssec:qsqd}, we provide some properties of the confidence of measurement in quantum sequence discrimination, and show that the maximum confidence of identifying a quantum sequence is to achieve the maximum confidence of identifying each state comprising the quantum sequence.
In Sec.~\ref{ssec:fomc}, we show that the optimal discrimination of a quantum sequence ensemble under MCMs can always be factorized into that of each individual ensemble.
In Sec.~\ref{sec:disc}, we summarize our results, discuss a useful application, and suggest a direction for future research.

\section{Maximum-confidence discrimination of quantum states}\label{sec:mcst}
For a finite-dimensional complex Hilbert space $\mathcal{H}$, let $\mathbb{H}$ be the set of all Hermitian operators acting on $\mathcal{H}$. 
We also denote by $\mathbb{H}_{+}$ the set of all positive-semidefinite operators in $\mathbb{H}$, that is,
\begin{equation}\label{eq:psdf}
\mathbb{H}_{+}=\{E\in\mathbb{H}\,|\,\bra{v}E\ket{v}\geqslant0~\mbox{for all}~\ket{v}\in\mathcal{H}\}.
\end{equation}
A quantum state is described by a density operator $\rho$, that is, $\rho\in\mathbb{H}_{+}$ with $\Tr\rho=1$.
A measurement is expressed by a positive operator-valued measure $\{M_{i}\}_{i}$, that is, $\{M_{i}\}_{i}\subseteq\mathbb{H}_{+}$ with $\sum_{i}M_{i}=\mathbbm{1}$ where $\mathbbm{1}$ is the identity operator in $\mathbb{H}$.
For a quantum state $\rho$, performing a measurement $\{M_{i}\}_{i}$ gives the measurement outcome corresponding to $M_{i}$ with probability $\Tr(\rho M_{i})$.
\subsection{Quantum state discrimination and confidence of measurement}\label{ssec:qsd}
Let us consider the situation of discriminating $n$ quantum states $\rho_{1},\ldots,\rho_{n}$ where the state $\rho_{i}$ is prepared with the \emph{nonzero} probability $\eta_{i}$ for each $i=1,\ldots,n$.
We denote this situation as an ensemble 
\begin{equation}\label{eq:esdf}
\mathcal{E}=\{\eta_{i},\rho_{i}\}_{i\in\mathbb{N}_{n}}
\end{equation}
where $\mathbb{N}_{n}$ is the set of all integers from $1$ to $n$, that is,
\begin{equation}\label{eq:nndf}
\mathbb{N}_{n}=\{1,\ldots,n\}.
\end{equation}
The state of a quantum system prepared from $\mathcal{E}$ is denoted by $\rho_{0}$, that is,
\begin{equation}\label{eq:rhzd}
\rho_{0}=\sum_{i\in\mathbb{N}_{n}}\eta_{i}\rho_{i}.
\end{equation}

To guess the prepared state from the ensemble $\mathcal{E}$ in Eq.~\eqref{eq:esdf}, we apply the decision rule using a measurement
\begin{equation}\label{eq:msdf}
\mathcal{M}=\{M_{?}\}\cup\{M_{i}\}_{i\in\mathbb{N}_{n}}.
\end{equation}
For each $i\in\mathbb{N}_{n}$, obtaining the measurement outcome corresponding to $M_{i}$ leads us to guess the prepared state as $\rho_{i}$.
On the other hand, the measurement outcome corresponding to $M_{?}$ is inconclusive: ``I don't know what state is prepared.''
In this case, the average probability of correctly guessing the prepared state is
\begin{equation}\label{eq:pcor}
\sum_{i\in\mathbb{N}_{n}}\eta_{i}\Tr(\rho_{i}M_{i}).
\end{equation}

Now, we introduce the confidence of a measurement in quantum state discrimination.
\begin{definition}\label{def:comd}
For an ensemble $\mathcal{E}=\{\eta_{i},\rho_{i}\}_{i\in\mathbb{N}_{n}}$, a measurement $\mathcal{M}=\{M_{?}\}\cup\{M_{i}\}_{i\in\mathbb{N}_{n}}$ and $x\in\mathbb{N}_{n}$, the \emph{confidence} of $\mathcal{M}$ to identify $\rho_{x}$ is the conditional probability
\begin{equation}\label{eq:cfdf}
\Pr(\rho_{x}|M_{x})=\frac{\eta_{x}\Tr(\rho_{x}M_{x})}{\Tr(\rho_{0}M_{x})},
\end{equation}
which is well defined only when
\begin{equation}\label{eq:cwdc}
\Tr(\rho_{0}M_{x})>0.
\end{equation}
The \emph{maximum confidence} of measurements to identify $\rho_{x}$ is
\begin{equation}\label{eq:dfmc}
\mathcal{C}_{x}(\mathcal{E})=\max_{\substack{\mathcal{M}~\mathsf{with}\\ \Tr(\rho_{0}M_{x})>0}}\Pr(\rho_{x}|M_{x}),
\end{equation}
where the maximum is taken over all possible measurements $\mathcal{M}=\{M_{?}\}\cup\{M_{i}\}_{i\in\mathbb{N}_{n}}$ with Eq.~\eqref{eq:cwdc}.
\end{definition}

Here, we note that the confidence $\Pr(\rho_{x}|M_{x})$ in Eq.~\eqref{eq:cfdf} represents the probability that the prepared state is $\rho_{x}$ given that the measurement outcome corresponding to $M_{x}$ is obtained.
We also note that the maximum confidence $\mathcal{C}_{x}(\mathcal{E})$ in Eq.~\eqref{eq:dfmc} is known to be the largest eigenvalue of $\sqrt{\rho_{0}}^{-1}\eta_{x}\rho_{x}\sqrt{\rho_{0}}^{-1}$ where $\sqrt{\rho_{0}}$ is the positive square root of $\rho_{0}$ and $\sqrt{\rho_{0}}^{-1}$ is the inverse of $\sqrt{\rho_{0}}$ on its support\cite{crok2006,herz2012,lee2022,supp}.
From this fact, we can easily see that
\begin{align}\label{eq:cxed}
\mathcal{C}_{x}(\mathcal{E})=&\min\{q\in\mathbb{R}\,|\,q\mathbbm{1}-\sqrt{\rho_{0}}^{-1}\eta_{x}\rho_{x}\sqrt{\rho_{0}}^{-1}\in\mathbb{H}_{+}\}\nonumber\\
=&\min\{q\in\mathbb{R}\,|\,q\rho_{0}-\eta_{x}\rho_{x}\in\mathbb{H}_{+}\}
\end{align}
where $\mathbb{R}$ is the set of all real numbers.

For an ensemble $\mathcal{E}=\{\eta_{i},\rho_{i}\}_{i\in\mathbb{N}_{n}}$, the following lemma shows that the maximum confidence $\mathcal{C}_{i}(\mathcal{E})$ is bounded below by the nonzero probability $\eta_{i}$ for all $i\in\mathbb{N}_{n}$.
\begin{lemma}\label{lem:lbmc}
For an ensemble $\mathcal{E}=\{\eta_{i},\rho_{i}\}_{i\in\mathbb{N}_{n}}$ and $x\in\mathbb{N}_{n}$, we have
\begin{equation}\label{eq:lbmc}
\mathcal{C}_{x}(\mathcal{E})=\max_{\substack{E\in\mathbb{H}_{+}\\ \Tr(\rho_{0}E)=1}}\eta_{x}\Tr(\rho_{x}E)\geqslant\eta_{x},
\end{equation}
where $\rho_{0}$ is defined in Eq.~\eqref{eq:rhzd} and the maximum is taken over all possible $E\in\mathbb{H}_{+}$ with $\Tr(\rho_{0}E)=1$.
\end{lemma}
\begin{proof}
From the definition of $\mathcal{C}_{x}(\mathcal{E})$ in Eq.~\eqref{eq:dfmc}, we have
\begin{equation}\label{eq:lbpv}
\mathcal{C}_{x}(\mathcal{E})=
\max_{\substack{M,\mathbbm{1}-M\in\mathbb{H}_{+}\\
\Tr(\rho_{0}M)>0}}\frac{\eta_{x}\Tr(\rho_{x}M)}{\Tr(\rho_{0}M)}
\leqslant\max_{\substack{E\in\mathbb{H}_{+}\\ \Tr(\rho_{0}E)=1}}\eta_{x}\Tr(\rho_{x}E).
\end{equation}
Therefore, to prove the equality in \eqref{eq:lbmc}, it is enough to show that
\begin{equation}\label{eq:ubpv}
\max_{\substack{M,\mathbbm{1}-M\in\mathbb{H}_{+}\\
\Tr(\rho_{0}M)>0}}\frac{\eta_{x}\Tr(\rho_{x}M)}{\Tr(\rho_{0}M)})\geqslant\max_{\substack{E\in\mathbb{H}_{+}\\ \Tr(\rho_{0}E)=1}}\eta_{x}\Tr(\rho_{x}E).
\end{equation}

For an optimal operator $E^{\star}$ providing the maximum in the right-hand side of Inequality~\eqref{eq:ubpv}, it follows from $E^{\star}\in\mathbb{H}_{+}$ and $\Tr(\rho_{0}E^{\star})=1$ that the Hermitian operator $M=E^{\star}/\Tr E^{\star}$ satisfies
\begin{equation}\label{eq:cxbp}
M,\mathbbm{1}-M\in\mathbb{H}_{+},~\Tr(\rho_{0}M)>0,~
\frac{\eta_{x}\Tr(\rho_{x}M)}{\Tr(\rho_{0}M)}=\eta_{x}\Tr(\rho_{x}E^{\star}),
\end{equation}
which lead us to Inequality~\eqref{eq:ubpv}.
Thus the equality in \eqref{eq:lbmc} is satisfied.
Moreover, the inequality in \eqref{eq:lbmc} holds because $E=\mathbbm{1}\in\mathbb{H}_{+}$ satisfies $\Tr(\rho_{0}E)=1$ and $\Tr(\rho_{x}E)=1$.
\qedhere
\end{proof}

Recently, it was shown that Eq.~\eqref{eq:lbmc} of Lemma~\ref{lem:lbmc} can also be derived in a different way by using Eq.~\eqref{eq:cxed}\cite{lee2022}.

\subsection{Optimal quantum state discrimination under MCMs}\label{ssec:omcd}
In this subsection, we first introduce the notion of MCM and present 
the definitions and properties related to the optimal quantum state discrimination under MCMs.
\begin{definition}\label{def:mcm}
For an ensemble $\mathcal{E}=\{\eta_{i},\rho_{i}\}_{i\in\mathbb{N}_{n}}$, a measurement $\{M_{?}\}\cup\{M_{i}\}_{i\in\mathbb{N}_{n}}$ is called a \emph{MCM} of $\mathcal{E}$ if it realizes the maximum confidence $\mathcal{C}_{i}(\mathcal{E})$ in Eq.~\eqref{eq:dfmc}, that is,
\begin{equation}\label{eq:dmcm}
\frac{\eta_{i}\Tr(\rho_{i}M_{i})}{\Tr(\rho_{0}M_{i})}=\mathcal{C}_{i}(\mathcal{E}),
\end{equation}
for all $i\in\mathbb{N}_{n}$ with $\Tr(\rho_{0}M_{i})>0$, or equivalently
\begin{equation}\label{eq:emcm}
\Tr[(\mathcal{C}_{i}(\mathcal{E})\rho_{0}-\eta_{i}\rho_{i})M_{i}]=0
\end{equation}
for all $i\in\mathbb{N}_{n}$.
\end{definition}

For a given ensemble $\mathcal{E}=\{\eta_{i},\rho_{i}\}_{i\in\mathbb{N}_{n}}$, we denote the set of all MCMs of $\mathcal{E}$ as
\begin{equation}\label{eq:smcm}
\mathbb{M}(\mathcal{E})=\{\mbox{measurement}~\{M_{?}\}\cup\{M_{i}\}_{i\in\mathbb{N}_{n}}\,|
\,M_{i}\in\mathbb{M}_{i}(\mathcal{E})~\mbox{for all}~i\in\mathbb{N}_{n}\}
\end{equation}
where $\mathbb{M}_{i}(\mathcal{E})$ is the set of all positive-semidefinite operators $M_{i}$ with Eq.~\eqref{eq:emcm}, that is,
\begin{equation}\label{eq:mxed}
\mathbb{M}_{i}(\mathcal{E})=\{E\in\mathbb{H}_{+}\,|\,\Tr[(\mathcal{C}_{i}(\mathcal{E})\rho_{0}-\eta_{i}\rho_{i})E]=0\}.
\end{equation}
For each $i\in\mathbb{N}_{n}$, the dual set of $\mathbb{M}_{i}(\mathcal{E})$ is defined as
\begin{equation}\label{eq:dset}
\mathbb{M}_{i}^{*}(\mathcal{E})=\{E\in\mathbb{H}\,|\,\Tr(EF)\geqslant0~\mbox{for all}~F\in\mathbb{M}_{i}(\mathcal{E})\}.
\end{equation}
Note that $\mathbb{M}_{i}(\mathcal{E})$ and $\mathbb{M}_{i}^{*}(\mathcal{E})$ are the dual sets of each other, since $\mathbb{M}_{i}(\mathcal{E})$ is closed and convex\cite{boyd2004}.
We also note that
\begin{equation}\label{eq:inre}
\mathbb{M}_{i}(\mathcal{E})\subseteq\mathbb{H}_{+}\subseteq\mathbb{M}_{i}^{*}(\mathcal{E})\subseteq\mathbb{H},
\end{equation}
where the first and last inclusions are from the definitions of $\mathbb{M}_{i}(\mathcal{E})$ and $\mathbb{M}_{i}^{*}(\mathcal{E})$ in Eqs.~\eqref{eq:mxed} and \eqref{eq:dset}, respectively, and the second inclusion is by $\Tr(EF)\geqslant0$ for all $E,F\in\mathbb{H}_{+}$.

The following lemma shows that $\mathbb{M}_{i}(\mathcal{E})$ and $\mathbb{M}_{i}^{*}(\mathcal{E})$ admit a representation in terms of projection operators.
\begin{lemma}\label{lem:mmpt}
For an ensemble $\mathcal{E}=\{\eta_{i},\rho_{i}\}_{i\in\mathbb{N}_{n}}$ and $x\in\mathbb{N}_{n}$, we have
\begin{subequations}\label{eq:mmpt}
\begin{align}
\mathbb{M}_{x}(\mathcal{E})=&\{E\in\mathbb{H}_{+}\,|\,E\Pi_{x}(\mathcal{E})=\mathbb{O}\},\label{eq:ompt}\\
\mathbb{M}_{x}^{*}(\mathcal{E})=&\{E\in\mathbb{H}\,|\,\Pi_{x}^{\bot}(\mathcal{E})E\Pi_{x}^{\bot}(\mathcal{E})\in\mathbb{H}_{+}\}\label{eq:dmpt}
\end{align}
\end{subequations}
where $\Pi_{x}(\mathcal{E})$ and $\Pi_{x}^{\bot}(\mathcal{E})$ are the projection operators onto the support and kernel of $\mathcal{C}_{x}(\mathcal{E})\rho_{0}-\eta_{x}\rho_{x}$, respectively, and $\mathbb{O}$ is the zero operator in $\mathbb{H}$.
\end{lemma}
\begin{proof}
Since Eq.~\eqref{eq:cxed} implies $\mathcal{C}_{x}(\mathcal{E})\rho_{0}-\eta_{x}\rho_{x}\in\mathbb{H}_{+}$, it follows from the definition of $\Pi_{x}(\mathcal{E})$ that $E\in\mathbb{H}_{+}$ is orthogonal to $\mathcal{C}_{x}(\mathcal{E})\rho_{0}-\eta_{x}\rho_{x}$  if and only if it is orthogonal to $\Pi_{x}(\mathcal{E})$. 
Thus the definition of $\mathbb{M}_{x}(\mathcal{E})$ in Eq.~\eqref{eq:mxed} leads us to Eq.~\eqref{eq:ompt}.

Equation~\eqref{eq:ompt} implies
\begin{equation}\label{eq:more}
\mathbb{M}_{x}(\mathcal{E})=\{\Pi_{x}^{\bot}(\mathcal{E})F\Pi_{x}^{\bot}(\mathcal{E})\,|\,F\in\mathbb{H}_{+}\}
\end{equation}
because $\Pi_{x}^{\bot}(\mathcal{E})E\Pi_{x}^{\bot}(\mathcal{E})=E$ for all $E\in\mathbb{H}_{+}$ orthogonal to $\Pi_{x}(\mathcal{E})$ and $\Pi_{x}^{\bot}(\mathcal{E})F\Pi_{x}^{\bot}(\mathcal{E})$ is positive semidefinite and orthogonal to $\Pi_{x}(\mathcal{E})$ for all $F\in\mathbb{H}_{+}$.
From the definition of $\mathbb{M}_{x}^{*}(\mathcal{E})$ in Eq.~\eqref{eq:dset} together with Eq.~\eqref{eq:more}, we have
\begin{align}\label{eq:mdor}
\mathbb{M}_{x}^{*}(\mathcal{E})&=\{E\in\mathbb{H}\,|\,\Tr[E\Pi_{x}^{\bot}(\mathcal{E})F\Pi_{x}^{\bot}(\mathcal{E})]\geqslant0
~\mbox{for all}~F\in\mathbb{H}_{+}\}\nonumber\\
&=\{E\in\mathbb{H}\,|\,\Tr[\Pi_{x}^{\bot}(\mathcal{E})E\Pi_{x}^{\bot}(\mathcal{E})F]\geqslant0
~\mbox{for all}~F\in\mathbb{H}_{+}\}.
\end{align}
Since $G\in\mathbb{H}_{+}$ if and only if $\Tr(GF)\geqslant0$ for all $F\in\mathbb{H}_{+}$, Eq.~\eqref{eq:mdor} leads us to Eq.~\eqref{eq:dmpt}.\qedhere
\end{proof}

Now, let us consider the optimal quantum state discrimination of an ensemble $\mathcal{E}=\{\eta_{i},\rho_{i}\}_{i\in\mathbb{N}_{n}}$ under MCMs, which is to achieve the maximum success probability
\begin{equation}\label{eq:pgdf}
p_{\sf G}(\mathcal{E})=\max_{\mathcal{M}\in\mathbb{M}(\mathcal{E})}\sum_{i\in\mathbb{N}_{n}}\eta_{i}\Tr(\rho_{i}M_{i})
\end{equation}
where the maximum is taken over all possible MCMs $\mathcal{M}=\{M_{?}\}\cup\{M_{i}\}_{i\in\mathbb{N}_{n}}$ of $\mathcal{E}$ in Eq.~\eqref{eq:smcm}.
The following theorem shows that $p_{\sf G}(\mathcal{E})$ in Eq.~\eqref{eq:pgdf} admits a strictly positive lower bound.
\begin{theorem}\label{thm:pglb}
For an ensemble $\mathcal{E}=\{\eta_{i},\rho_{i}\}_{i\in\mathbb{N}_{n}}$, we have
\begin{equation}\label{eq:pglb}
p_{\sf G}(\mathcal{E})\geqslant\frac{\lambda}{n}
\end{equation}
where $\lambda$ is the smallest nonzero eigenvalue of $\rho_{0}$ in Eq.~\eqref{eq:rhzd}.
\end{theorem}
\begin{proof}
For each $i\in\mathbb{N}_{n}$, Lemma~\ref{lem:lbmc} ensures the existence of $E_{i}\in\mathbb{H}_{+}$ satisfying
\begin{equation}\label{eq:fite}
\Tr(\rho_{0}E_{i})=1,~
\eta_{i}\Tr(\rho_{i}E_{i})=\mathcal{C}_{i}(\mathcal{E}).
\end{equation}
By denoting the projection operator onto the support of $\rho_{0}$ as $\Pi(\mathcal{E})$, we have
\begin{eqnarray}\label{eq:desp}
1&=&\Tr(\rho_{0}E_{i})\nonumber\\
&=&\Tr[\rho_{0}\Pi(\mathcal{E})E_{i}\Pi(\mathcal{E})]\nonumber\\
&\leqslant&\Tr[\Pi(\mathcal{E})E_{i}\Pi(\mathcal{E})]\nonumber\\
&\leqslant&\Tr[\tfrac{1}{\lambda}\rho_{0}\Pi(\mathcal{E})E_{i}\Pi(\mathcal{E})]\nonumber\\
&=&\tfrac{1}{\lambda}\Tr(\rho_{0}E_{i})=\tfrac{1}{\lambda},
\end{eqnarray}
where the first and last equalities are by Eq.~\eqref{eq:fite}, the second and third equalities are from the definition of $\Pi(\mathcal{E})$, and the inequalities are due to $\Pi(\mathcal{E})-\rho_{0}\in\mathbb{H}_{+}$ and $\rho_{0}/\lambda-\Pi(\mathcal{E})\in\mathbb{H}_{+}$.

By using $E_{1},\ldots,E_{n}$ and $\Pi(\mathcal{E})$, we can define a measurement $\mathcal{M}=\{M_{?}\}\cup\{M_{i}\}_{i\in\mathbb{N}_{n}}$ with
\begin{subequations}\label{eq:mmmc}
\begin{align}
M_{?}&=\mathbbm{1}-\frac{\sum_{i'\in\mathbb{N}_{n}}\Pi(\mathcal{E})E_{i'}\Pi(\mathcal{E})}{\sum_{j\in\mathbb{N}_{n}}\Tr[\Pi(\mathcal{E})E_{j}\Pi(\mathcal{E})]},\label{eq:mmqc}\\
M_{i}&=\frac{\Pi(\mathcal{E})E_{i}\Pi(\mathcal{E})}{\sum_{j\in\mathbb{N}_{n}}\Tr[\Pi(\mathcal{E})E_{j}\Pi(\mathcal{E})]}\label{eq:mmic}
\end{align}
\end{subequations}
for all $i\in\mathbb{N}_{n}$.
This measurement qualifies as a MCM of $\mathcal{E}$ in Eq.~\eqref{eq:dmcm} because Eqs.~\eqref{eq:fite} and \eqref{eq:mmmc} together with Inequality~\eqref{eq:desp} imply
\begin{subequations}\label{eq:vcct}
\begin{align}
\Tr(\rho_{0}M_{i})&\geqslant\tfrac{\lambda}{n}>0,\label{eq:vccf}\\
\frac{\eta_{i}\Tr(\rho_{i}M_{i})}{\Tr(\rho_{0}M_{i})}&=\mathcal{C}_{i}(\mathcal{E})\label{eq:vccs}
\end{align}
\end{subequations}
for all $i\in\mathbb{N}_{n}$.

For the measurement $\mathcal{M}$ in Eq.~\eqref{eq:mmmc}, we have
\begin{eqnarray}\label{eq:plbu}
p_{\sf G}(\mathcal{E})&\geqslant&\sum_{i\in\mathbb{N}_{n}}\eta_{i}\Tr(\rho_{i}M_{i})\nonumber\\
&=&\frac{\sum_{i\in\mathbb{N}_{n}}\mathcal{C}_{i}(\mathcal{E})}{\sum_{j\in\mathbb{N}_{n}}\Tr[\Pi(\mathcal{E})E_{j}\Pi(\mathcal{E})]}\nonumber\\
&\geqslant&\frac{\sum_{i\in\mathbb{N}_{n}}\eta_{i}}{\sum_{j\in\mathbb{N}_{n}}\Tr[\Pi(\mathcal{E})E_{j}\Pi(\mathcal{E})]}\geqslant\frac{\lambda}{n},
\end{eqnarray}
where the first inequality is by the definition of $p_{\sf G}(\mathcal{E})$ in Eq.~\eqref{eq:pgdf}, the first equality is from Eqs.~\eqref{eq:fite} and \eqref{eq:mmmc} together with the definition of $\Pi(\mathcal{E})$, the second inequality is from Lemma~\ref{lem:lbmc}, and the last inequality is due to Inequality~\eqref{eq:desp}.
Thus our theorem is true.\qedhere
\end{proof}

For a given ensemble $\mathcal{E}=\{\eta_{i},\rho_{i}\}_{i\in\mathbb{N}_{n}}$, we define $q_{\sf G}(\mathcal{E})$ as the minimum quantity
\begin{equation}\label{eq:qgbd}
q_{\sf G}(\mathcal{E})=\min\Tr H
\end{equation}
over all possible positive-semidefinite operators $H$ satisfying
\begin{equation}\label{eq:dudm}
H-\eta_{i}\rho_{i}\in\mathbb{M}_{i}^{*}(\mathcal{E})
\end{equation}
for all $i\in\mathbb{N}_{n}$, where $\mathbb{M}_{i}^{*}(\mathcal{E})$ is defined in Eq.~\eqref{eq:dset}.
The feasible set of operators in Eq.~\eqref{eq:qgbd} is thus given by
\begin{equation}\label{eq:hpbd}
\mathbb{H}(\mathcal{E})=\{H\in\mathbb{H}_{+}\,|\,H-\eta_{i}\rho_{i}\in\mathbb{M}_{i}^{*}(\mathcal{E})
~\mbox{for all}~i\in\mathbb{N}_{n}\}.
\end{equation}

The following theorem shows that $q_{\sf G}(\mathcal{E})$ in Eq.~\eqref{eq:qgbd} is equal to $p_{\sf G}(\mathcal{E})$ in Eq.~\eqref{eq:pgdf}.
The proof of Theorem~\ref{thm:pqge} is given in Appendix~\ref{app:pqge}.
\begin{theorem}\label{thm:pqge}
For an ensemble $\mathcal{E}=\{\eta_{i},\rho_{i}\}_{i\in\mathbb{N}_{n}}$, we have
\begin{equation}\label{eq:pqge}
p_{\sf G}(\mathcal{E})= q_{\sf G}(\mathcal{E}).
\end{equation}
\end{theorem}

For an ensemble $\mathcal{E}=\{\eta_{i},\rho_{i}\}_{i\in\mathbb{N}_{n}}$, the following theorem establishes a necessary and sufficient condition for $\mathcal{M}\in\mathbb{M}(\mathcal{E})$ and $H\in\mathbb{H}(\mathcal{E})$ to realize $p_{\sf G}(\mathcal{E})$ and $q_{\sf G}(\mathcal{E})$, respectively.
\begin{theorem}\label{thm:nspq}
For an ensemble $\mathcal{E}=\{\eta_{i},\rho_{i}\}_{i\in\mathbb{N}_{n}}$, $\mathcal{M}=\{M_{?}\}\cup\{M_{i}\}_{i\in\mathbb{N}_{n}}\in\mathbb{M}(\mathcal{E})$ and $H\in\mathbb{H}(\mathcal{E})$, $\mathcal{M}$ provides $p_{\sf G}(\mathcal{E})$ and $H$ realizes $q_{\sf G}(\mathcal{E})$ if and only if
\begin{equation}\label{eq:nspg}
\Tr(M_{?}H)=0,~
\Tr[M_{i}(H-\eta_{i}\rho_{i})]=0
\end{equation}
for all $i\in\mathbb{N}_{n}$.
\end{theorem}
\begin{proof}
Let us first suppose that $\mathcal{M}$ provides $p_{\sf G}(\mathcal{E})$ and $H$ gives $q_{\sf G}(\mathcal{E})$. 
From $M_{?}\in\mathbb{H}_{+}$ and $H\in\mathbb{H}_{+}$ together with $M_{i}\in\mathbb{M}_{i}(\mathcal{E})$ and $H-\eta_{i}\rho_{i}\in\mathbb{M}_{i}^{*}(\mathcal{E})$ for all $i\in\mathbb{N}_{n}$, we have
\begin{equation}\label{eq:tmhp}
\Tr(M_{?}H)\geqslant0,~
\Tr[M_{i}(H-\eta_{i}\rho_{i})]\geqslant0
\end{equation}
for all $i\in\mathbb{N}_{n}$.
We note that
\begin{eqnarray}\label{eq:stmh}
\Tr(M_{?}H)+\sum_{i\in\mathbb{N}_{n}}\Tr[M_{i}(H-\eta_{i}\rho_{i})]
&=&\Tr H-\sum_{i\in\mathbb{N}_{n}}\eta_{i}\Tr(\rho_{i}M_{i})\nonumber\\
&=&q_{\sf G}(\mathcal{E})-p_{\sf G}(\mathcal{E})=0,
\end{eqnarray}
where the first equality is by $M_{?}+\sum_{i\in\mathbb{N}_{n}}M_{i}=\mathbbm{1}$, the second equality is from the assumption of $\mathcal{M}$ and $H$, and the last equality is from Theorem~\ref{thm:pqge}.
Inequalities~\eqref{eq:tmhp} and Eq.~\eqref{eq:stmh} lead us to Condition~\eqref{eq:nspg}.

Conversely, let us assume that $\mathcal{M}$ and $H$ satisfy Condition~\eqref{eq:nspg}. From this assumption, we have
\begin{eqnarray}\label{eq:qgeq}
q_{\sf G}(\mathcal{E})&=&p_{\sf G}(\mathcal{E})\nonumber\\
&\geqslant&\sum_{i\in\mathbb{N}_{n}}\eta_{i}\Tr(\rho_{i}M_{i})\nonumber\\
&=&\sum_{i\in\mathbb{N}_{n}}\eta_{i}\Tr(\rho_{i}M_{i})+\Tr(M_{?}H)+\sum_{i\in\mathbb{N}_{n}}\Tr[M_{i}(H-\eta_{i}\rho_{i})]\nonumber\\
&=&\Tr H\geqslant q_{\sf G}(\mathcal{E}),
\end{eqnarray}
where the first equality is from Theorem~\ref{thm:pqge}, the first inequality is from the definition of $p_{\sf G}(\mathcal{E})$ in Eq.~\eqref{eq:pgdf}, the second equality is by Condition~\eqref{eq:nspg}, the last equality is due to $M_{?}+\sum_{i\in\mathbb{N}_{n}}M_{i}=\mathbbm{1}$, and the second inequality is from the definition of $q_{\sf G}(\mathcal{E})$ in Eq.~\eqref{eq:qgbd}.
Inequality~\eqref{eq:qgeq} leads us to
\begin{equation}\label{eq:mhpq}
\sum_{i\in\mathbb{N}_{n}}\eta_{i}\Tr(\rho_{i}M_{i})=p_{\sf G}(\mathcal{E}),~
\Tr H=q_{\sf G}(\mathcal{E}).
\end{equation}
Thus $\mathcal{M}$ provides $p_{\sf G}(\mathcal{E})$ and $H$ realizes $q_{\sf G}(\mathcal{E})$.
\qedhere
\end{proof}

We also note that the definition of $q_{\sf G}(\mathcal{E})$ in Eq.~\eqref{eq:qgbd}, as well as its relationship with $p_{\sf G}(\mathcal{E})$ in Theorems~\ref{thm:pqge} and \ref{thm:nspq}, can be derived within the framework of generalized quantum state discrimination\cite{naka2015}.
In this sense, Theorems~\ref{thm:pqge} and \ref{thm:nspq} can be regarded as a special case of generalized quantum state discrimination when the available measurements are restricted to MCMs\cite{naka2015}.

The following lemma shows that the support of any Hermitian operator realizing $q_{\sf G}(\mathcal{E})$ is contained in that of $\rho_{0}$.
\begin{lemma}\label{lem:nwmh}
For an ensemble $\mathcal{E}=\{\eta_{i},\rho_{i}\}_{i\in\mathbb{N}_{n}}$ and $H\in\mathbb{H}(\mathcal{E})$ realizing $q_{\sf G}(\mathcal{E})$, we have
\begin{equation}\label{eq:hpeh}
\Pi(\mathcal{E}) H\Pi(\mathcal{E})=H
\end{equation}
where $\Pi(\mathcal{E})$ is the projection operator onto the support of $\rho_{0}$ in Eq.~\eqref{eq:rhzd}.
\end{lemma}
\begin{proof}
For a MCM $\mathcal{M}=\{M_{?}\}\cup\{M_{i}\}_{i\in\mathbb{N}_{n}}$ of $\mathcal{E}$ providing $p_{\sf G}(\mathcal{E})$, we construct another measurement $\mathcal{M}'=\{M_{?}'\}\cup\{M_{i}'\}_{i\in\mathbb{N}_{n}}$ as
\begin{equation}\label{eq:micd}
M_{?}'=\Pi(\mathcal{E}) M_{?}\Pi(\mathcal{E})+\Pi^{\bot}(\mathcal{E}),~
M_{i}'=\Pi(\mathcal{E}) M_{i}\Pi(\mathcal{E})
\end{equation}
for all $i\in\mathbb{N}_{n}$, where $\Pi^{\bot}(\mathcal{E})$ is the projection operator onto the kernel of $\rho_{0}$.
We first show that $\mathcal{M}'$ is a MCM of $\mathcal{E}$ providing $p_{\sf G}(\mathcal{E})$ and then prove Eq.~\eqref{eq:hpeh} by using this property.

For each $i\in\mathbb{N}_{n}$, we have
\begin{equation}\label{eq:mipm}
M_{i}'\in\mathbb{M}_{i}(\mathcal{E})
\end{equation}
because
\begin{eqnarray}\label{eq:crmi}
\Tr[(\mathcal{C}_{i}(\mathcal{E})\rho_{0}-\eta_{i}\rho_{i})M_{i}']
&=&\Tr[\Pi(\mathcal{E})(\mathcal{C}_{i}(\mathcal{E})\rho_{0}-\eta_{i}\rho_{i})\Pi(\mathcal{E}) M_{i}]
\nonumber\\
&=&\Tr[(\mathcal{C}_{i}(\mathcal{E})\rho_{0}-\eta_{i}\rho_{i})M_{i}]=0,
\end{eqnarray}
where the first equality is by Eq.~\eqref{eq:micd}, the second equality is from the definition of $\Pi(\mathcal{E})$, and the last equality is due to $M_{i}\in\mathbb{M}_{i}(\mathcal{E})$.
Thus $\mathcal{M}'$ is a MCM of $\mathcal{E}$.
Moreover, it provides $p_{\sf G}(\mathcal{E})$ because
\begin{align}\label{eq:mppg}
\sum_{i\in\mathbb{N}_{n}}\eta_{i}\Tr(\rho_{i}M_{i}')
&=\sum_{i\in\mathbb{N}_{n}}\eta_{i}\Tr(\rho_{i}M_{i})=p_{\sf G}(\mathcal{E}),
\end{align}
where the first equality is from Eq.~\eqref{eq:micd} together with the definition of $\Pi(\mathcal{E})$ and the second equality is from the assumption of $\mathcal{M}$.

Now, we show that Eq.~\eqref{eq:hpeh} is satisfied.
The positive-semidefinite operator $H$ is orthogonal to $\Pi^{\bot}(\mathcal{E})$ because
\begin{eqnarray}\label{eq:zqz}
0&=&\Tr(M_{?}'H)\nonumber\\
&=&\Tr\big([\Pi(\mathcal{E}) M_{?}\Pi(\mathcal{E})+\Pi^{\bot}(\mathcal{E})]H\big)\nonumber\\
&\geqslant&\Tr\big(\Pi^{\bot}(\mathcal{E})H\big)\geqslant0,
\end{eqnarray}
where the first equality is from Theorem~\ref{thm:nspq}, the second equality is by Eq.~\eqref{eq:micd}, and the inequalities are by the fact that $\Pi(\mathcal{E}) M_{?}\Pi(\mathcal{E})$, $\Pi^{\bot}(\mathcal{E})$ and $H$ are in $\mathbb{H}_{+}$.
Thus we have
\begin{equation}\label{eq:fowh}
\Pi(\mathcal{E})H\Pi(\mathcal{E})=\big[\mathbbm{1}-\Pi^{\bot}(\mathcal{E})\big]H\big[\mathbbm{1}-\Pi^{\bot}(\mathcal{E})\big]=H,
\end{equation}
where the first and second equalities are due to $\Pi(\mathcal{E})+\Pi^{\bot}(\mathcal{E})=\mathbbm{1}$ and the orthogonality between $H$ and $\Pi^{\bot}(\mathcal{E})$, respectively.\qedhere
\end{proof}

\section{Maximum-confidence discrimination of quantum sequences}\label{sec:mcse}
\subsection{Quantum sequence discrimination and confidence of measurement}\label{ssec:qsqd}
Let us consider the situation of discriminating quantum sequences, each being a sequence of quantum states prepared from $L$ quantum state ensembles
\begin{equation}\label{eq:qses}
\mathcal{E}^{1}=\{\eta_{i}^{1},\rho_{i}^{1}\}_{i\in\mathbb{N}_{n_{1}}},\ldots,
\mathcal{E}^{L}=\{\eta_{i}^{L},\rho_{i}^{L}\}_{i\in\mathbb{N}_{n_{L}}},
\end{equation}
where each ensemble $\mathcal{E}^{l}$ consists of $n_{l}$ states $\rho_{1}^{l},\ldots,\rho_{n_{l}}^{l}$ prepared with probabilities $\eta_{1}^{l},\ldots,\eta_{n_{l}}^{l}$, for $l=1,\ldots,L$.
In this situation, the preparation of a quantum sequence proceeds step by step. At the first step, a state $\rho_{c_{1}}^{1}$ is prepared from the ensemble $\mathcal{E}^{1}$ with the corresponding probability $\eta_{c_{1}}^{1}$.
At each subsequent step $l=2,\ldots,L$, a state $\rho_{c_{l}}^{l}$ is prepared from the ensemble $\mathcal{E}^{l}$ with the corresponding probability $\eta_{c_{l}}^{l}$. Consequently, the quantum sequence
\begin{equation}\label{eq:qscl}
(\rho_{c_{1}}^{1},\ldots,\rho_{c_{L}}^{L})
\end{equation}
is prepared with the probability
\begin{equation}\label{eq:etcl}
\eta_{c_{1}}^{1}\times\cdots\times\eta_{c_{L}}^{L}.
\end{equation}

By using the notion
\begin{equation}\label{eq:vecn}
\vec{n}=(n_{1},\ldots,n_{L}),
\end{equation}
we denote the Cartesian product of $\mathbb{N}_{n_{1}},\ldots,\mathbb{N}_{n_{L}}$ as
\begin{equation}\label{eq:nvcn}
\mathbb{N}_{\vec{n}}=
\{(c_{1},\ldots,c_{L})\,|\,c_{1}\in\mathbb{N}_{n_{1}},\ldots,c_{L}\in\mathbb{N}_{n_{L}}\}.
\end{equation}
The quantum sequence in Eq.~\eqref{eq:qscl}, together with its probability in Eq.~\eqref{eq:etcl}, can be considered as the tensor-producted state $\rho_{\vec{c}}$ with the joint probability $\eta_{\vec{c}}$,
\begin{equation}\label{eq:revc}
\eta_{\vec{c}}=\prod_{l=1}^{L}\eta_{c_{l}}^{l},~
\rho_{\vec{c}}=\bigotimes_{l=1}^{L}\rho_{c_{l}}^{l},
\end{equation}
for $\vec{c}=(c_{1},\ldots,c_{L})\in\mathbb{N}_{\vec{n}}$, from the tensor product of the ensembles in Eq.~\eqref{eq:qses},
\begin{equation}\label{eq:bote}
\bigotimes_{l=1}^{L}\mathcal{E}^{l}=\{\eta_{\vec{c}},\rho_{\vec{c}}\}_{\vec{c}\in\mathbb{N}_{\vec{n}}}.
\end{equation}
Figure~\ref{fig:qseb} illustrates the process of preparing a quantum sequence from the quantum sequence ensemble in Eq.~\eqref{eq:bote}.

\begin{figure}[!tt]
\centerline{\includegraphics{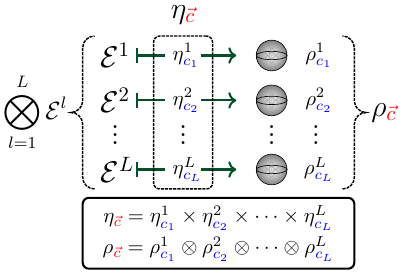}}
\caption{Quantum sequence ensemble $\bigotimes_{l=1}^{L}\mathcal{E}^{l}=\{\eta_{\vec{c}},\rho_{\vec{c}}\}_{\vec{c}\in\mathbb{N}_{\vec{n}}}$.
For each $l=1,\ldots,L$, the $l$th state $\rho_{c_{l}}^{l}$ is prepared from the ensemble $\mathcal{E}^{l}$ with the probability $\eta_{c_{l}}^{l}$, for $c_{l}=1,\ldots,n_{l}$.
The quantum sequence $(\rho_{c_{1}}^{1},\ldots,\rho_{c_{L}}^{L})$ can be regarded as the tensor producted state $\rho_{\vec{c}}$ with the probability $\eta_{\vec{c}}$, for $\vec{c}=(c_{1},\ldots,c_{L})\in\mathbb{N}_{\vec{n}}$. 
}\label{fig:qseb}
\end{figure}

The quantum state of a quantum system from the ensemble in Eq.~\eqref{eq:bote} is denoted by $\rho_{0}$, that is,
\begin{equation}\label{eq:qsqs}
\rho_{0}=\sum_{\vec{c}\in\mathbb{N}_{\vec{n}}}\eta_{\vec{c}}\rho_{\vec{c}}.
\end{equation}
From Eq.~\eqref{eq:revc}, the state $\rho_{0}$ in Eq.~\eqref{eq:qsqs} can be rewritten as
\begin{equation}\label{eq:aqsd}
\rho_{0}=\!\sum_{c_{1}\in\mathbb{N}_{n_{1}}}\!\cdots\!\sum_{c_{L}\in\mathbb{N}_{n_{L}}}
\!\eta_{c_{1}}^{1}\rho_{c_{1}}^{1}\otimes\cdots\otimes\eta_{c_{L}}^{L}\rho_{c_{L}}^{L}
=\bigotimes_{l=1}^{L}\rho_{0}^{l},
\end{equation}
where $\rho_{0}^{l}$ is the quantum state of a quantum system from the ensemble $\mathcal{E}^{l}$, as defined in Eq.~\eqref{eq:rhzd}, for $l=1,\ldots,L$.

To discriminate the quantum sequences from the ensemble $\mathcal{E}=\bigotimes_{l=1}^{L}\mathcal{E}^{l}$ in Eq.~\eqref{eq:bote}, we use a measurement
\begin{equation}\label{eq:sems}
\mathcal{M}=\{M_{?}\}\cup\{M_{\vec{c}}\}_{\vec{c}\in\mathbb{N}_{\vec{n}}}
\end{equation}
with the decision rule that the measurement outcome corresponding to $M_{\vec{c}}$ leads us to guess the prepared quantum sequence as $\rho_{\vec{c}}$, for $\vec{c}\in\mathbb{N}_{\vec{n}}$, and $M_{?}$ provides inconclusive measurement outcomes.
In this case, the average probability of correctly guessing the prepared quantum sequence from $\mathcal{E}$ is
\begin{equation}\label{eq:seqp}
\sum_{\vec{c}\in\mathbb{N}_{\vec{n}}}\eta_{\vec{c}}\Tr(\rho_{\vec{c}}M_{\vec{c}}).
\end{equation}

In particular, when a measurement $\mathcal{M}^{l}=\{M_{?}^{l}\}\cup\{M_{i}^{l}\}_{i\in\mathbb{N}_{n_{l}}}$ is independently performed at each step $l=1,\ldots,L$ of a quantum sequence from $\mathcal{E}$, the measurement operator $M_{\vec{c}}$ in Eq.~\eqref{eq:sems} can be represented as
\begin{equation}\label{eq:mcrp}
M_{\vec{c}}=M_{c_{1}}^{1}\otimes\cdots\otimes M_{c_{L}}^{L}
\end{equation}
for each $\vec{c}=(c_{1},\ldots,c_{L})\in\mathbb{N}_{\vec{n}}$; therefore, we have
\begin{equation}\label{eq:sprp}
\sum_{\vec{c}\in\mathbb{N}_{\vec{n}}}\eta_{\vec{c}}\Tr(\rho_{\vec{c}}M_{\vec{c}})=
\prod_{l=1}^{L}\sum_{i\in\mathbb{N}_{n_{l}}}\eta_{i}^{l}\Tr(\rho_{i}^{l}M_{i}^{l}).
\end{equation}
In other words, Eq.~\eqref{eq:seqp} can be expressed as the product of average success probabilities for discriminating each ensemble $\mathcal{E}^{l}$ with respect to the measurement $\mathcal{M}^{l}$, for $l=1,\ldots,L$.

Now, let us consider the confidence of a measurement $\mathcal{M}=\{M_{?}\}\cup\{M_{\vec{c}}\}_{\vec{c}\in\mathbb{N}_{\vec{n}}}$ in quantum sequence discrimination of the ensemble $\mathcal{E}=\bigotimes_{l=1}^{L}\mathcal{E}^{l}$ in Eq.~\eqref{eq:bote}. 
For each $\vec{x}\in\mathbb{N}_{n}$, the \emph{confidence} of $\mathcal{M}$ to identify $\rho_{\vec{x}}$ is the probability $\Pr(\rho_{\vec{x}}|M_{\vec{x}})$ that the prepared quantum sequence is $\rho_{\vec{x}}$ given that the measurement outcome corresponding to $M_{\vec{x}}$ is obtained, that is,
\begin{equation}\label{eq:cfds}
\Pr(\rho_{\vec{x}}|M_{\vec{x}})=\frac{\eta_{\vec{x}}\Tr(\rho_{\vec{x}}M_{\vec{x}})}{\Tr(\rho_{0}M_{\vec{x}})},
\end{equation}
which is well defined only when
\begin{equation}\label{eq:cwds}
\Tr(\rho_{0}M_{\vec{x}})>0.
\end{equation}
The \emph{maximum confidence} of measurements to identify $\rho_{\vec{x}}$ is
\begin{equation}\label{eq:dfms}
\mathcal{C}_{\vec{x}}(\mathcal{E})=\max_{\substack{\mathcal{M}~\mathsf{with}\\ \Tr(\rho_{0}M_{\vec{x}})>0}}\Pr(\rho_{\vec{x}}|M_{\vec{x}}),
\end{equation}
where the maximum is taken over all possible measurements $\mathcal{M}=\{M_{?}\}\cup\{M_{\vec{c}}\}_{\vec{c}\in\mathbb{N}_{\vec{n}}}$ with Eq.~\eqref{eq:sems}.

The following theorem shows that the maximum confidence of identifying a quantum sequence is to achieve the maximum confidence of identifying each state in the quantum sequence.
\begin{theorem}\label{thm:fmcr}
For a quantum sequence ensemble $\mathcal{E}=\bigotimes_{l=1}^{L}\mathcal{E}^{l}$ in Eq.~\eqref{eq:bote} and $\vec{x}=(x_{1},\ldots,x_{L})\in\mathbb{N}_{\vec{n}}$, we have
\begin{equation}\label{eq:fmcr}
\mathcal{C}_{\vec{x}}(\mathcal{E})=\prod_{l=1}^{L}\mathcal{C}_{x_{l}}(\mathcal{E}^{l}),
\end{equation}
where $\mathcal{C}_{x_{l}}(\mathcal{E}^{l})$ is the maximum confidence to identify $\rho_{x_{l}}^{l}$ from the ensemble $\mathcal{E}^{l}$, as defined in Eq.~\eqref{eq:dfmc}, for $l=1,\ldots,L$.
\end{theorem}
To prove Theorem~\ref{thm:fmcr}, we employ the following lemma that provides an expression for the sum and difference between two tensor-product operators.
The proof of Lemma~\ref{lem:wxyz} is given in Appendix~\ref{app:wxyz}.
\begin{lemma}\label{lem:wxyz}
For tensor-producted operators $X=\bigotimes_{l=1}^{L}X_{l}$ and $Y=\bigotimes_{l=1}^{L}Y_{l}$, we have
\begin{equation}\label{eq:xmwy}
X+(-1)^{a}Y=\frac{1}{2^{L-1}}\sum_{\substack{\vec{b}\in\mathbb{Z}_{2}^{L}\\ \omega_{2}(\vec{b})=a}}\bigotimes_{l=1}^{L}\big[X_{l}+(-1)^{b_{l}}Y_{l}\big]
\end{equation}
for each $a\in\{0,1\}$, where $\mathbb{Z}_{2}^{L}$ is the Cartesian product of $L$ copies of $\{0,1\}$, $b_{l}$ is the $l$th entry of $\vec{b}=(b_{1},\ldots,b_{L})$, and $\omega_{2}(\vec{b})$ is the modulo 2 summation of all entries in $\vec{b}$, that is,
\begin{equation}\label{eq:omts}
\omega_{2}(\vec{b})=\sum_{l=1}^{L}b_{l}~(\mathrm{mod}~2).
\end{equation}
\end{lemma}

\begin{proof}[Proof of Theorem~\ref{thm:fmcr}]
We prove Eq.~\eqref{eq:fmcr} by showing that
\begin{subequations}\label{eq:bmce}
\begin{align}
\mathcal{C}_{\vec{x}}(\mathcal{E})\geqslant\prod_{l=1}^{L}\mathcal{C}_{x_{l}}(\mathcal{E}^{l}),
\label{eq:lmce}\\
\mathcal{C}_{\vec{x}}(\mathcal{E})\leqslant\prod_{l=1}^{L}\mathcal{C}_{x_{l}}(\mathcal{E}^{l}).
\label{eq:rmce}
\end{align}
\end{subequations}

To show Inequality~\eqref{eq:lmce}, let $\mathcal{M}^{l}=\{M_{?}^{l}\}\cup\{M_{i}^{l}\}_{i\in\mathbb{N}_{n_{l}}}$ be a measurement providing the maximum-confidence $\mathcal{C}_{x_{l}}(\mathcal{E}^{l})$ in Eq.~\eqref{eq:dfmc} for each $l=1,\ldots,L$, that is,
\begin{subequations}\label{eq:emmc}
\begin{align}
\Tr(\rho_{0}^{l}M_{x_{l}}^{l})&>0,\label{eq:emmi}\\
\frac{\eta_{x_{l}}^{l}\Tr(\rho_{x_{l}}^{l}M_{x_{l}}^{l})}{\Tr(\rho_{0}^{l}M_{x_{l}}^{l})}&=\mathcal{C}_{x_{l}}(\mathcal{E}^{l}).\label{eq:emme}
\end{align}
\end{subequations}
For the measurement $\{M_{?}\}\cup\{M_{\vec{c}}\}_{\vec{c}\in\mathbb{N}_{\vec{n}}}$ constructed as Eq.~\eqref{eq:mcrp} in terms of $\mathcal{M}^{l}$ realizing the maximum-confidence $\mathcal{C}_{x_{l}}(\mathcal{E}^{l})$ for $l=1,\ldots,L$, we have
\begin{subequations}\label{eq:vmcg}
\begin{gather}
\Tr(\rho_{0}M_{\vec{x}})=\prod_{l=1}^{L}\Tr(\rho_{0}^{l}M_{x_{l}}^{l})>0,
\label{eq:vmci}\\
\frac{\eta_{\vec{x}}\Tr(\rho_{\vec{x}}M_{\vec{x}})}{\Tr(\rho_{0}M_{\vec{x}})}=\prod_{l=1}^{L}\frac{\eta_{x_{l}}^{l}\Tr(\rho_{x_{l}}^{l}M_{x_{l}}^{l})}{\Tr(\rho_{0}^{l}M_{x_{l}}^{l})}
=\prod_{l=1}^{L}\mathcal{C}_{x_{l}}(\mathcal{E}^{l}).
\label{eq:vmce}
\end{gather}
\end{subequations}
The first equalities in Eqs.~\eqref{eq:vmci} and \eqref{eq:vmce} follow from Eqs.~\eqref{eq:revc} and \eqref{eq:mcrp}.
The inequality in Eq.~\eqref{eq:vmci} and the second equality in Eq.~\eqref{eq:vmce} are obtained from Inequality~\eqref{eq:emmi} and Eq.~\eqref{eq:emme}, respectively.
Thus the definition of $\mathcal{C}_{\vec{x}}(\mathcal{E})$ in Eq.~\eqref{eq:dfms} leads us to Inequality~\eqref{eq:lmce}.

From Lemma~\ref{lem:wxyz} together with Eqs.~\eqref{eq:revc} and \eqref{eq:aqsd}, we have
\begin{equation}\label{eq:crmr}
\left[\prod_{l=1}^{L}\mathcal{C}_{x_{l}}(\mathcal{E}^{l})\right]\rho_{0}-\eta_{\vec{x}}\rho_{\vec{x}}
=\frac{1}{2^{L-1}}\sum_{\substack{\vec{b}\in\mathbb{Z}_{2}^{L}\\ \omega_{2}(\vec{b})=1}}\bigotimes_{l=1}^{L}\Big[\mathcal{C}_{x_{l}}(\mathcal{E}^{l})\rho_{0}^{l}+(-1)^{b_{l}}\eta_{x_{l}}^{l}\rho_{x_{l}}^{l}\Big],
\end{equation}
where $b_{l}$ is the $l$th entry of $\vec{b}=(b_{1},\ldots,b_{L})$.
The right-hand side of Eq.~\eqref{eq:crmr} is positive semidefinite because Eq.~\eqref{eq:cxed} implies that $\mathcal{C}_{x_{l}}(\mathcal{E}^{l})\rho_{0}^{l}-\eta_{x_{l}}^{l}\rho_{x_{l}}^{l}$ is positive semidefinite for all $l=1,\ldots,L$.
Thus, Eq.~\eqref{eq:cxed} leads us to Inequality~\eqref{eq:rmce}.\qedhere
\end{proof}

\subsection{Optimal quantum sequence discrimination under MCMs}\label{ssec:fomc}
In this subsection, we present our main result, showing that the optimal quantum sequence discrimination under MCMs can always be achieved just by performing quantum state discrimination independently at each step of the quantum sequence.
To this end, we first introduce the notion of MCM in quantum sequence discrimination.

For a quantum sequence ensemble $\mathcal{E}=\bigotimes_{l=1}^{L}\mathcal{E}^{l}$ in Eq.~\eqref{eq:bote} and a measurement $\mathcal{M}=\{M_{?}\}\cup\{M_{\vec{c}}\}_{\vec{c}\in\mathbb{N}_{\vec{n}}}$, $\mathcal{M}$ is a \emph{MCM} of $\mathcal{E}$ if and only if it realizes the maximum confidence $\mathcal{C}_{\vec{c}}(\mathcal{E})$ in Eq.~\eqref{eq:dfms}, that is,
\begin{equation}\label{eq:mcms}
\frac{\eta_{\vec{c}}\Tr(\rho_{\vec{c}}M_{\vec{c}})}{\Tr(\rho_{0}M_{\vec{c}})}=\mathcal{C}_{\vec{c}}(\mathcal{E}),
\end{equation}
for all $\vec{c}\in\mathbb{N}_{\vec{n}}$ with $\Tr(\rho_{0}M_{\vec{c}})>0$, or equivalently
\begin{equation}\label{eq:sqcm}
\Tr[(\mathcal{C}_{\vec{c}}(\mathcal{E})\rho_{0}-\eta_{\vec{c}}\rho_{\vec{c}})M_{\vec{c}}]=0
\end{equation}
for all $\vec{c}\in\mathbb{N}_{\vec{n}}$.
Similarly to Eq.~\eqref{eq:smcm}, we denote by $\mathbb{M}(\mathcal{E})$ the set of all MCMs of $\mathcal{E}$, that is,
\begin{equation}\label{eq:qsmc}
\mathbb{M}(\mathcal{E})=\{\mbox{measurement}~\{M_{?}\}\cup\{M_{\vec{c}}\}_{\vec{c}\in\mathbb{N}_{\vec{n}}}\,|\,M_{\vec{c}}\in\mathbb{M}_{\vec{c}}(\mathcal{E})~\mbox{for all}~\vec{c}\in\mathbb{N}_{\vec{n}}\}
\end{equation}
where $\mathbb{M}_{\vec{c}}(\mathcal{E})$ is the set of all positive-semidefinite operators $M_{\vec{c}}$ with Eq.~\eqref{eq:sqcm}, that is,
\begin{equation}\label{eq:qmxe}
\mathbb{M}_{\vec{c}}(\mathcal{E})=\{E\in\mathbb{H}_{+}\,|\,\Tr[(\mathcal{C}_{\vec{c}}(\mathcal{E})\rho_{0}-\eta_{\vec{c}}\rho_{\vec{c}})E]=0\}.
\end{equation}

The following theorem shows that performing a MCM independently at each step of the quantum sequence is a MCM for the quantum sequence.
\begin{theorem}\label{eq:pmie}
When $\mathcal{M}^{l}=\{M_{?}^{l}\}\cup\{M_{i}^{l}\}_{i\in\mathbb{N}_{n_{l}}}$ is a MCM of $\mathcal{E}^{l}=\{\eta_{i}^{l},\rho_{i}^{l}\}_{i\in\mathbb{N}_{n_{l}}}$ for each $l=1,\ldots,L$, the measurement $\mathcal{M}=\{M_{?}\}\cup\{M_{\vec{c}}\}_{\vec{c}\in\mathbb{N}_{\vec{n}}}$ satisfying Eq.~\eqref{eq:mcrp} for all $\vec{c}\in\mathbb{N}_{\vec{n}}$ is a MCM of the quantum sequence ensemble $\mathcal{E}=\bigotimes_{l=1}^{L}\mathcal{E}^{l}$ in Eq.~\eqref{eq:bote}.
\end{theorem}
\begin{proof}
As $\mathcal{M}$ is obviously a measurement, it is enough to show the validity of Eq.~\eqref{eq:sqcm} for all $\vec{c}\in\mathbb{N}_{\vec{n}}$.

For each $\vec{x}=(x_{1},\ldots,x_{L})\in\mathbb{N}_{\vec{n}}$, it follows from Theorem~\ref{thm:fmcr} and Lemma~\ref{lem:wxyz} that
\begin{equation}\label{eq:crps}
\mathcal{C}_{\vec{x}}(\mathcal{E})\rho_{0}-\eta_{\vec{x}}\rho_{\vec{x}}
=\frac{1}{2^{L-1}}\sum_{\substack{\vec{b}\in\mathbb{Z}_{2}^{L}\\ \omega_{2}(\vec{b})=1}}\bigotimes_{l=1}^{L}\Big[\mathcal{C}_{x_{l}}(\mathcal{E}^{l})\rho_{0}^{l}+(-1)^{b_{l}}\eta_{x_{l}}^{l}\rho_{x_{l}}^{l}\Big],
\end{equation}
where $b_{l}$ is the $l$th entry of $\vec{b}=(b_{1},\ldots,b_{L})$.
From Eqs.~\eqref{eq:mcrp} and \eqref{eq:crps}, we have
\begin{equation}\label{eq:tcrp}
\Tr\big[\big(\mathcal{C}_{\vec{x}}(\mathcal{E})\rho_{0}-\eta_{\vec{x}}\rho_{\vec{x}}\big)M_{\vec{x}}\big]
=\frac{1}{2^{L-1}}\sum_{\substack{\vec{b}\in\mathbb{Z}_{2}^{L}\\ \omega_{2}(\vec{b})=1}}\bigotimes_{l=1}^{L}\Tr\big[\big(\mathcal{C}_{x_{l}}(\mathcal{E}^{l})\rho_{0}^{l}
+(-1)^{b_{l}}\eta_{x_{l}}^{l}\rho_{x_{l}}^{l}\big)M_{x_{l}}\big].
\end{equation}
Since every $\vec{b}\in\mathbb{Z}_{2}^{L}$ with $\omega_{2}(\vec{b})=1$ has at least one entry equal to $1$ and $M_{x_{l}}^{l}$ is an element of $\mathbb{M}_{x_{l}}(\mathcal{E}^{l})$ defined in Eq.~\eqref{eq:mxed} for all $l=1,\ldots,L$, the right-hand side of Eq.~\eqref{eq:tcrp} vanishes; therefore, Eq.~\eqref{eq:sqcm} holds.\qedhere
\end{proof}

Now, let us consider the optimal discrimination of a quantum sequence ensemble $\mathcal{E}=\bigotimes_{l=1}^{L}\mathcal{E}^{l}$ under MCMs to achieve
\begin{equation}\label{eq:pgsd}
p_{\sf G}(\mathcal{E})=\max_{\mathcal{M}\in\mathbb{M}(\mathcal{E})}\sum_{\vec{c}\in\mathbb{N}_{\vec{n}}}\eta_{\vec{c}}\Tr(\rho_{\vec{c}}M_{\vec{c}})
\end{equation}
where the maximum is taken over all possible MCMs $\mathcal{M}=\{M_{?}\}\cup\{M_{\vec{c}}\}_{\vec{c}\in\mathbb{N}_{\vec{n}}}$ of $\mathcal{E}$ in Eq.~\eqref{eq:qsmc}.
From Theorem~\ref{eq:pmie} together with Eq.~\eqref{eq:sprp}, $p_{\sf G}(\mathcal{E})$ in Eq.~\eqref{eq:pgsd} is bounded below by the product of the product of the maximum success probabilities $p_{\sf G}(\mathcal{E}^{l})$ for discriminating each ensemble $\mathcal{E}^{l}$, for $l=1,\ldots,L$, that is,
\begin{equation}\label{eq:pgpp}
p_{\sf G}(\mathcal{E})\geqslant\prod_{l=1}^{L}p_{\sf G}(\mathcal{E}^{l}).
\end{equation}

The following theorem shows that Inequality~\eqref{eq:pgpp} is saturated.
In other words, the optimal quantum sequence discrimination in Eq.~\eqref{eq:pgsd} can always be factorized into the optimal quantum state discrimination for each individual ensemble.
\begin{theorem}\label{thm:pgpg}
For a quantum sequence ensemble $\mathcal{E}=\bigotimes_{l=1}^{L}\mathcal{E}^{l}$ in Eq.~\eqref{eq:bote}, we have
\begin{equation}\label{eq:pepp}
p_{\sf G}(\mathcal{E})=\prod_{l=1}^{L}p_{\sf G}(\mathcal{E}^{l}).
\end{equation}
\end{theorem}
To prove Theorem~\ref{thm:pgpg}, we use the following lemma that establishes a necessary condition satisfied by every element of the set in Eq.~\eqref{eq:qmxe}.
The proof of Lemma~\ref{lem:ppee} is given in Appendix~\ref{app:ppee}.
\begin{lemma}\label{lem:ppee}
For a quantum sequence ensemble $\mathcal{E}=\bigotimes_{l=1}^{L}\mathcal{E}^{l}$ in Eq.~\eqref{eq:bote}, $\vec{x}=(x_{1},\ldots,x_{L})\in\mathbb{N}_{\vec{n}}$ and $E\in\mathbb{M}_{\vec{x}}(\mathcal{E})$, we have
\begin{equation}\label{eq:peps}
\Bigg[\bigotimes_{l=1}^{L}\Pi_{x_{l}}^{\bot}(\mathcal{E}^{l})\Bigg]
\Pi(\mathcal{E})E\Pi(\mathcal{E})
\Bigg[\bigotimes_{l=1}^{L}\Pi_{x_{l}}^{\bot}(\mathcal{E}^{l})\Bigg]
=\Pi(\mathcal{E})E\Pi(\mathcal{E})
\end{equation}
where $\Pi(\mathcal{E})$ is the projection operator onto the support of $\rho_{0}$ and $\Pi_{x_{l}}^{\bot}(\mathcal{E}^{l})$ is the projection operator onto the kernel of $\mathcal{C}_{x_{l}}(\mathcal{E}^{l})\rho_{0}^{l}-\eta_{x_{l}}^{l}\rho_{x_{l}}^{l}$, for $l=1,\ldots,L$.
\end{lemma}

For convenience, we briefly sketch the proof of Theorem~\ref{thm:pgpg} and clarify the role of Lemma~\ref{lem:ppee}.
We consider a Hermitian operator $H=\bigotimes_{l=1}^{L}H^{l}$, where $H^{1},\ldots,H^{L}$ are Hermitian operators achieving $q_{\sf G}(\mathcal{E}^{1}),\ldots,q_{\sf G}(\mathcal{E}^{L})$, respectively.
Lemma~\ref{lem:ppee}, together with Lemmas~\ref{lem:mmpt}--\ref{lem:wxyz}, is then used to show that $H\in\mathbb{H}(\mathcal{E})$, which implies the bound $q_{\sf G}(\mathcal{E})\leqslant\Tr H$.
Combining this bound with Theorem~\ref{thm:pqge} establishes the converse of Inequality~\eqref{eq:pgpp}.

\begin{proof}[Proof of Theorem~\ref{thm:pgpg}]
As we already have Inequality~\eqref{eq:pgpp}, it is enough to show that
\begin{equation}\label{eq:pleg}
p_{\sf G}(\mathcal{E})\leqslant\prod_{l=1}^{L}p_{\sf G}(\mathcal{E}^{l}).
\end{equation}

For each $l=1,\ldots,L$, let us suppose that $H^{l}\in\mathbb{H}(\mathcal{E}^{l})$ realizes $q_{\sf G}(\mathcal{E}^{l})$, where $q_{\sf G}(\mathcal{E}^{l})$ and $\mathbb{H}(\mathcal{E}^{l})$ are defined in Eqs.~\eqref{eq:qgbd} and \eqref{eq:hpbd}, respectively.
For the Hermitian operator 
\begin{equation}\label{eq:teph}
H=\bigotimes_{l=1}^{L}H^{l},
\end{equation}
we show $H\in\mathbb{H}(\mathcal{E})$, or equivalently
\begin{equation}\label{eq:hinh}
H-\eta_{\vec{c}}\rho_{\vec{c}}\in\mathbb{M}_{\vec{c}}^{*}(\mathcal{E})
\end{equation}
for all $\vec{c}\in\mathbb{N}_{\vec{n}}$; therefore, we have
\begin{equation}\label{eq:purn}
p_{\sf G}(\mathcal{E})=q_{\sf G}(\mathcal{E})
\leqslant\Tr H
=\prod_{l=1}^{L}q_{\sf G}(\mathcal{E}^{l})
=\prod_{l=1}^{L}p_{\sf G}(\mathcal{E}^{l}),
\end{equation}
where the first and last equalities are from Theorem~\ref{thm:pqge}, the inequality follows from the definition of $q_{\sf G}(\mathcal{E})$ in Eq.~\eqref{eq:qgbd}, the second equality is from Eq.~\eqref{eq:teph}, and the assumption of $H^{l}$ for $l=1,\ldots,L$.

For each $\vec{c}\in\mathbb{N}_{\vec{n}}$, every $E\in\mathbb{M}_{\vec{x}}(\mathcal{E})$ satisfies
\begin{align}\label{eq:tehe}
\Tr[E(H-\eta_{\vec{x}}\rho_{\vec{x}})]
&=\frac{1}{2^{L-1}}\sum_{\substack{\vec{b}\in\mathbb{Z}_{2}^{L}\\ \omega_{2}(\vec{b})=1}}\Tr\Bigg[\Pi(\mathcal{E})E\Pi(\mathcal{E})\bigotimes_{l=1}^{L}\big(H^{l}+(-1)^{b_{l}}\eta_{c_{l}}^{l}\rho_{c_{l}}^{l}\big)\Bigg]\nonumber\\
&=\frac{1}{2^{L-1}}\sum_{\substack{\vec{b}\in\mathbb{Z}_{2}^{L}\\ \omega_{2}(\vec{b})=1}}\Tr\Bigg[
\Big[\bigotimes_{l=1}^{L}\Pi_{c_{l}}^{\bot}(\mathcal{E}^{l})\Big]
\Pi(\mathcal{E})E\Pi(\mathcal{E})
\Big[\bigotimes_{l=1}^{L}\Pi_{c_{l}}^{\bot}(\mathcal{E}^{l})\Big]
\bigotimes_{l=1}^{L}\big(H^{l}+(-1)^{b_{l}}\eta_{c_{l}}^{l}\rho_{c_{l}}^{l}\big)\Bigg]\nonumber\\
&=\frac{1}{2^{L-1}}\sum_{\substack{\vec{b}\in\mathbb{Z}_{2}^{L}\\ \omega_{2}(\vec{b})=1}}\Tr\Bigg[
\Pi(\mathcal{E})E\Pi(\mathcal{E})\bigotimes_{l=1}^{L}
\Pi_{c_{l}}^{\bot}(\mathcal{E}^{l})
\big(H^{l}+(-1)^{b_{l}}\eta_{c_{l}}^{l}\rho_{c_{l}}^{l}\big)
\Pi_{c_{l}}^{\bot}(\mathcal{E}^{l})
\Bigg],
\end{align}
where $\Pi(\mathcal{E})$ is the projection operator onto the support of $\rho_{0}$, and $\Pi(\mathcal{E}^{l})$ and $\Pi_{c_{l}}^{\bot}(\mathcal{E}^{l})$ are the projection operators onto the support of $\rho_{0}^{l}$ and the kernel of $\mathcal{C}_{c_{l}}(\mathcal{E}^{l})\rho_{0}^{l}-\eta_{c_{l}}^{l}\rho_{c_{l}}^{l}$, respectively, for $l=1,\ldots,L$.
In Eq.~\eqref{eq:tehe}, the first equality is from Lemmas~\ref{lem:nwmh} and \ref{lem:wxyz} together with $\Pi(\mathcal{E})=\bigotimes_{l=1}^{L}\Pi(\mathcal{E}^{l})$ and the second equality is from Lemma~\ref{lem:ppee}.
The last term in Eq.~\eqref{eq:tehe} is non-negative because $E\in\mathbb{M}_{\vec{c}}(\mathcal{E})$ ensures the positive semidefiniteness of $\Pi(\mathcal{E})E\Pi(\mathcal{E})$ and $H^{l}-\eta_{c_{l}}^{l}\rho_{c_{l}}^{l}\in\mathbb{M}_{\vec{c}}^{*}(\mathcal{E}^{l})$, together with Lemma~\ref{lem:mmpt}, implies the positive semidefiniteness of $\Pi_{c_{l}}^{\bot}(\mathcal{E}^{l})\big(H^{l}\pm\eta_{c_{l}}^{l}\rho_{c_{l}}^{l}\big)\Pi_{c_{l}}^{\bot}(\mathcal{E}^{l})$ for each $l=1,\ldots,L$.
Thus the definition of $\mathbb{M}_{\vec{c}}^{*}(\mathcal{E})$ in Eq.~\eqref{eq:dset} leads us to Inclusion~\eqref{eq:hinh}.\qedhere
\end{proof}

We illustrate Theorem~\ref{thm:pgpg} in the following example.
\begin{example}\label{ex:fmcm}
For a positive integer $L$, let us consider the quantum sequence ensembles $\mathcal{E}=\bigotimes_{l=1}^{L}\mathcal{E}^{l}$, where each $\mathcal{E}^{l}$ is the qubit state ensemble $\{\eta_{i}^{l},\rho_{i}^{l}\}_{i=1}^{3}$ with
\begin{align}\label{eq:ex}
\eta_{1}^{l}=\tfrac{1}{3},~
&\rho_{1}^{l}=\ket{\psi_{1}^{l}}\!\bra{\psi_{1}^{l}},~
\ket{\psi_{1}^{l}}=\cos\theta_{l}\ket{0}+\sin\theta_{l}\ket{1},\nonumber\\
\eta_{2}^{l}=\tfrac{1}{3},~
&\rho_{2}^{l}=\ket{\psi_{2}^{l}}\!\bra{\psi_{2}^{l}},~
\ket{\psi_{2}^{l}}=\cos\theta_{l}\ket{0}+e^{\frac{2\pi\mathsf{i}}{3}}\sin\theta_{l}\ket{1},
\nonumber\\
\eta_{3}^{l}=\tfrac{1}{3},~
&\rho_{3}^{l}=\ket{\psi_{3}^{l}}\!\bra{\psi_{3}^{l}},~
\ket{\psi_{3}^{l}}=\cos\theta_{l}\ket{0}+e^{\frac{4\pi\mathsf{i}}{3}}\sin\theta_{l}\ket{1},
~0<\theta_{l}\leqslant\tfrac{\pi}{4},
\end{align}
for $l=1,\ldots,L$.
\end{example}

For each $l=1,\ldots,L$, the measurement $\mathcal{M}^{l}=\{M_{?}^{l}\}\cup\{M_{i}^{l}\}_{i=1}^{3}$, given by
\begin{align}\label{eq:exmcm}
M_{?}^{l}&=(1-\tan^{2}\theta_{l})\ket{0}\!\bra{0},\nonumber\\
M_{1}^{l}&=\tfrac{1}{3\cos^{2}\theta_{l}}\ket{\phi_{1}^{l}}\!\bra{\phi_{1}^{l}},~\ket{\phi_{1}^{l}}=\sin\theta_{l}\ket{0}+\cos\theta_{l}\ket{1},\nonumber\\
M_{2}^{l}&=\tfrac{1}{3\cos^{2}\theta_{l}}\ket{\phi_{2}^{l}}\!\bra{\phi_{2}^{l}},~\ket{\phi_{2}^{l}}=\sin\theta_{l}\ket{0}+e^{\frac{2\pi\mathsf{i}}{3}}\cos\theta_{l}\ket{1},\nonumber\\
M_{3}^{l}&=\tfrac{1}{3\cos^{2}\theta_{l}}\ket{\phi_{3}^{l}}\!\bra{\phi_{3}^{l}},~\ket{\phi_{3}^{l}}=\sin\theta_{l}\ket{0}+e^{\frac{4\pi\mathsf{i}}{3}}\cos\theta_{l}\ket{1},
\end{align}
is the MCM of $\mathcal{E}^{l}$ providing $p_{\sf G}(\mathcal{E}^{l})$\cite{crok2006}, that is,
\begin{equation}\label{eq:cpgv}
p_{\sf G}(\mathcal{E}^{l})=\sum_{i=1}^{3}\eta_{i}^{l}\Tr(\rho_{i}^{l}M_{i}^{l})=\frac{4}{3}\sin^{2}\theta_{l}.
\end{equation}
Thus Theorem~\ref{thm:pgpg} leads us to
\begin{equation}\label{eq:rst}
p_{\sf G}(\mathcal{E})=\prod_{l=1}^{L}p_{\sf G}(\mathcal{E}^{l})=\prod_{l=1}^{L}\Big(\frac{4}{3}\sin^{2}\theta_{l}\Big).
\end{equation}
In other words, the optimal quantum sequence discrimination of $\mathcal{E}$ under MCMs can be achieved just by performing the MCM $\mathcal{M}^{l}$ in Eq.~\eqref{eq:exmcm} independently at each step $l=1,\ldots,L$ of the quantum sequence.

\section{Discussion}\label{sec:disc}
We have considered quantum sequence discrimination under MCMs and shown that the optimal discrimination of a quantum sequence ensemble can always be factorized into that of each individual ensemble.
In other words, the optimal quantum sequence discrimination under MCMs can be achieved just by performing a maximum-confidence discrimination independently at each step of the quantum sequence.
We have also shown that the maximum confidence of identifying a quantum sequence is to achieve the maximum confidence of identifying each state in the quantum sequence.
We have further established a necessary and sufficient condition for the optimal quantum state discrimination under MCMs.

We note that our results include the known results on unambiguous discrimination of quantum sequences\cite{gupt20241,gupt20242}.
In unambiguous discrimination of quantum states, the states are discriminated  without error (i.e., unambiguously) by allowing an inconclusive measurement outcome\cite{ivan1987,diek1988,pere1988}.
For a given quantum state ensemble, a state can be unambiguously discriminated if and only if the maximum confidence of identifying the state is equal to one\cite{crok2006}.
Therefore, it follows from Theorem~\ref{thm:fmcr} that a quantum sequence from a quantum sequence ensemble can be unambiguously discriminated if and only if each state in the quantum sequence can be unambiguously discriminated.

When quantum states can be unambiguously discriminated, MCMs become unambiguous measurements that provide unit confidence for each state.
In this case, the optimal state discrimination under MCMs coincides with the optimal unambiguous discrimination.
Therefore, Theorem~\ref{thm:pgpg} ensures that the optimal unambiguous discrimination of a quantum sequence ensemble can always be factorized into that of each individual ensemble when quantum sequences can be unambiguously discriminated.

Our results can be applied to situations involving multiple independent repetitions of quantum information processing tasks based on MCMs, such as quantum cryptography and quantum teleportation\cite{benn1992,neve2012}. 
In these situations, our results indicate that no advantage can be gained even if the receiver has access to quantum memory and considers performing a collective measurement on the quantum sequence.

As our results show the factorizability of the optimal quantum sequence discrimination under MCMs, it is natural to investigate measurement conditions under which the optimal discrimination of a quantum sequence ensemble can be factorized into that of each individual ensemble.
It is also an interesting direction for future research to investigate factorizability in the confidence of identifying a quantum sequence from a multiparty quantum sequence ensemble.

In particular, our analysis relies on the assumption that each state composing a quantum sequence is independently prepared from a given quantum state ensemble.
Relaxing this assumption to allow correlations or entanglement between different states composing the sequence may fundamentally change the structure of optimal discrimination strategies and potentially invalidate the factorizability property.
Moreover, the availability of quantum memory or adaptive feedback may enable collective measurement strategies that surpass the performance of factorizable ones under the same optimality criterion.
Investigating such scenarios remains an important open problem for future work.

\section*{Acknowledgments}
This work was supported by Korea Research Institute for defense Technology planning and advancement (KRIT) grant funded by Defense Acquisition Program Administration(DAPA)(KRIT-CT-23–031), a National Research Foundation of Korea(NRF) grant funded by the Korean government(Ministry of Science and ICT) (No.NRF2023R1A2C1007039), and the Institute for Information \& Communications Technology Planning \& Evaluation(IITP) grant funded by the Korean government(MSIP)(Grant No. RS-2025-02304540). JSK was supported by Creation of the Quantum Information Science R\&D Ecosystem(Grant No. 2022M3H3A106307411) through the National Research Foundation of Korea(NRF) funded by the Korean government(Ministry of Science and ICT).

\appendix
\appsection{Proof of Theorem~\ref{thm:pqge}}\label{app:pqge}
In this section, we prove Theorem~\ref{thm:pqge} by showing that
\begin{subequations}\label{eq:pgqg}
\begin{align}
p_{\sf G}(\mathcal{E})&\leqslant q_{\sf G}(\mathcal{E}),\label{eq:pqgu}\\
p_{\sf G}(\mathcal{E})&\geqslant q_{\sf G}(\mathcal{E}).\label{eq:pqgl}
\end{align}
\end{subequations}
\begin{proof}[Proof of Inequality~\eqref{eq:pqgu}]
Let us assume that $\mathcal{M}=\{M_{?}\}\cup\{M_{i}\}_{i\in\mathbb{N}_{n}}\in\mathbb{M}(\mathcal{E})$ and $H\in\mathbb{H}(\mathcal{E})$ realize $p_{\sf G}(\mathcal{E})$ and $q_{\sf G}(\mathcal{E})$, respectively.
From the definitions of $\mathbb{M}(\mathcal{E})$ and $\mathbb{H}(\mathcal{E})$ in Eqs.~\eqref{eq:smcm} and \eqref{eq:hpbd}, we have
\begin{equation}\label{eq:ocpc}
\Tr(M_{?}H)\geqslant0,~
\Tr[M_{i}(H-\eta_{i}\rho_{i})]\geqslant0
\end{equation}
for all $i\in\mathbb{N}_{n}$.
Thus we have
\begin{eqnarray}\label{eq:ppqg}
p_{\sf G}(\mathcal{E})&=&\sum_{i\in\mathbb{N}_{n}}\eta_{i}\Tr(\rho_{i}M_{i})\nonumber\\
&\leqslant&\sum_{i\in\mathbb{N}_{n}}\eta_{i}\Tr(\rho_{i}M_{i})
+\Tr(M_{?}H)
+\sum_{i\in\mathbb{N}_{n}}\Tr[M_{i}(H-\eta_{i}\rho_{i})]\nonumber\\
&=&\Tr H=q_{\sf G}(\mathcal{E}),
\end{eqnarray}
where the first equality is from the assumption of $\mathcal{M}$, the inequality is by Inequalities~\eqref{eq:ocpc}, the second equality is due to $M_{?}+\sum_{i\in\mathbb{N}_{n}}M_{i}=\mathbbm{1}$, and the last equality is from the assumption of $H$.\qedhere
\end{proof}
\begin{proof}[Proof of Inequality~\eqref{eq:pqgl}]
Let us consider the set
\begin{equation}\label{eq:sbdf}
\mathcal{S}(\mathcal{E})=\Big\{\Big(\sum_{i\in\mathbb{N}_{n}}\eta_{i}\Tr(\rho_{i}M_{i})-p,\,\mathbbm{1}-M_{?}-\sum_{i\in\mathbb{N}_{n}}M_{i}\Big)\in\mathbb{R}\times\mathbb{H}\,\Big|\,p>p_{\sf G}(\mathcal{E}),\,M_{?}\in\mathbb{H}_{+},
\,M_{1}\in\mathbb{M}_{1}(\mathcal{E}),\ldots,M_{n}\in\mathbb{M}_{n}(\mathcal{E})\Big\}.
\end{equation}
We note that $\mathcal{S}(\mathcal{E})$ is a convex set due to the convexity of $\mathbb{H}_{+}$ and $\mathbb{M}_{1}(\mathcal{E}),\ldots,\mathbb{M}_{n}(\mathcal{E})$.
Moreover, $\mathcal{S}(\mathcal{E})$ does not have the origin $(0,\mathbb{O})$ of $\mathbb{R}\times\mathbb{H}$; otherwise, there exists $\{M_{?}\}\cup\{M_{i}\}_{i\in\mathbb{N}_{n}}\in\mathbb{M}(\mathcal{E})$ with
\begin{equation}\label{eq:ctop}
\sum_{i\in\mathbb{N}_{n}}\eta_{i}\Tr(\rho_{i}M_{i})>p_{\sf G}(\mathcal{E}),
\end{equation}
and this contradicts the optimality of $p_{\sf G}(\mathcal{E})$ in Eq.~\eqref{eq:pgdf}.
We also note that the Cartesian product $\mathbb{R}\times\mathbb{H}$ can be considered as a real vector space with an inner product defined as
\begin{equation}\label{eq:indf}
\braket{(a,A),(b,B)}=ab+\Tr(AB)
\end{equation}
for $(a,A),(b,B)\in\mathbb{R}\times\mathbb{H}$.

Since $\mathcal{S}(\mathcal{E})$ and the single-element set $\{(0,\mathbb{O})\}$ are disjoint convex sets, it follows from the separating hyperplane theorem\cite{boyd2004,sht} that there exists $(\gamma,\Gamma)\in\mathbb{R}\times\mathbb{H}$ satisfying
\begin{subequations}\label{eq:ggoo}
\begin{gather}
(\gamma,\Gamma)\neq(0,\mathbb{O}),\label{eq:ggfc}\\
\braket{(\gamma,\Gamma),(r,R)}\leqslant0\label{eq:ggsc}
\end{gather}
\end{subequations}
for all $(r,R)\in\mathcal{S}(\mathcal{E})$.

Suppose
\begin{subequations}\label{eq:acob}
\begin{gather}
\Tr\Gamma\leqslant\gamma p_{\sf G}(\mathcal{E}),\label{eq:fcob}\\
\Gamma\in\mathbb{H}_{+},\label{eq:scob}\\
\gamma>0,\label{eq:tcob}\\ 
\Gamma-\gamma\eta_{i}\rho_{i}\in\mathbb{M}_{i}^{*}(\mathcal{E})\label{eq:lcob}
\end{gather}
\end{subequations}
for all $i\in\mathbb{N}_{n}$.
From Conditions~\eqref{eq:scob}, \eqref{eq:tcob} and \eqref{eq:lcob}, $\Gamma/\gamma$ is an element of $\mathbb{H}(\mathcal{E})$ in Eq.~\eqref{eq:hpbd}.
Thus the definition of $q_{\sf G}(\mathcal{E})$ in Eq.~\eqref{eq:qgbd} leads us to 
\begin{equation}\label{eq:srth}
q_{\sf G}(\mathcal{E})\leqslant\Tr(\Gamma/\gamma).
\end{equation}
Moreover, Conditions~\eqref{eq:fcob} and \eqref{eq:tcob} imply
\begin{equation}\label{eq:srlb}
\Tr(\Gamma/\gamma)\leqslant p_{\sf G}(\mathcal{E}).
\end{equation}
Inequalities~\eqref{eq:srth} and \eqref{eq:srlb} complete the proof of Inequality~\eqref{eq:pqgl}.
The rest of this proof is to prove Condition~\eqref{eq:acob}.\qedhere
\end{proof}

\begin{proof}[Proof of \eqref{eq:fcob}]
From Eq.~\eqref{eq:indf}, Inequality~\eqref{eq:ggsc} can be rewritten as
\begin{equation}\label{eq:inrt}
\Tr\Gamma-\Tr(M_{?}\Gamma)-\sum_{i\in\mathbb{N}_{n}}\Tr[M_{i}(\Gamma-\gamma\eta_{i}\rho_{i})]\leqslant\gamma p
\end{equation}
for all $p>p_{\sf G}(\mathcal{E})$ and all $\{M_{?}\}\cup\{M_{i}\}_{i\in\mathbb{N}_{n}}$ with
\begin{equation}\label{eq:domh}
M_{?}\in\mathbb{H}_{+},\,M_{1}\in\mathbb{M}_{1}(\mathcal{E}),\ldots,M_{n}\in\mathbb{M}_{n}(\mathcal{E}).
\end{equation}
$\{M_{?}\}\cup\{M_{i}\}_{i\in\mathbb{N}_{n}}$ with $M_{i}=\mathbb{O}$ for all $i\in\{?\}\cup\mathbb{N}_{n}$ clearly satisfies Inclusion~\eqref{eq:domh}; therefore, Inequality~\eqref{eq:inrt} becomes Inequality~\eqref{eq:fcob} by taking the limit of $p$ to $p_{\sf G}(\mathcal{E})$.\qedhere
\end{proof}

\begin{proof}[Proof of \eqref{eq:scob}]
For an arbitrary $M_{?}\in\mathbb{H}_{+}$ and $M_{i}=\mathbb{O}$ for all $i\in\mathbb{N}_{n}$, $\{M_{?}\}\cup\{M_{i}\}_{i\in\mathbb{N}_{n}}$ clearly satisfies Inclusion~\eqref{eq:domh}.
In this case, Inequality~\eqref{eq:inrt} becomes 
\begin{equation}\label{eq:ssif}
\Tr\Gamma-\Tr(M_{?}\Gamma)\leqslant\gamma p_{\sf G}(\mathcal{E})
\end{equation}
by taking the limit of $p$ to $p_{\sf G}(\mathcal{E})$.

Suppose $\Gamma\notin\mathbb{H}_{+}$; then there exists $M\in\mathbb{H}_{+}$ with $\Tr(M\Gamma)<0$. 
We note that $M\in\mathbb{H}_{+}$ implies $tM\in\mathbb{H}_{+}$ for all $t>0$.
Thus $\{M_{?}\}\cup\{M_{i}\}_{i\in\mathbb{N}_{n}}$ with $M_{?}=tM$ for $t>0$ and $M_{i}=\mathbb{O}$ for all $i\in\mathbb{N}_{n}$ also satisfies Inclusion~\eqref{eq:domh}.
Now, Inequality~\eqref{eq:ssif} can be rewritten as
\begin{equation}\label{eq:asif}
\Tr\Gamma-t\Tr(M\Gamma)\leqslant\gamma p_{\sf G}(\mathcal{E})
\end{equation}
Since Inequality~\eqref{eq:asif} is true for arbitrarily large $t>0$, $\gamma p_{\sf G}(\mathcal{E})$ can also be arbitrarily large.
However, this contradicts that both $\gamma$ and $p_{\sf G}(\mathcal{E})$ are finite.
Thus $\Gamma\in\mathbb{H}_{+}$, which completes the proof of Inequality~\eqref{eq:scob}.\qedhere
\end{proof}

\begin{proof}[Proof of \eqref{eq:tcob}]
Suppose $\Gamma\neq\mathbb{O}$.
From Inequality~\eqref{eq:scob}, we have $\Tr\Gamma>0$. 
Thus, Inequality~\eqref{eq:fcob} and Theorem~\ref{thm:pglb} lead us to $\gamma>0$.

Now, suppose $\Gamma=\mathbb{O}$. We have $\gamma\neq0$, otherwise a contradiction to Eq.~\eqref{eq:ggfc}.
The strict positivity of $\gamma$ follows from Inequality~\eqref{eq:fcob} and Theorem~\ref{thm:pglb} together with $\Tr\Gamma=0$ and $\gamma\neq0$.
Thus Inequality~\eqref{eq:tcob} holds regardless of $\Gamma$.\qedhere
\end{proof}

\begin{proof}[Proof of \eqref{eq:lcob}]
The proof method is analog to that of \eqref{eq:scob}.
For each $x\in\mathbb{N}_{n}$, let us consider an arbitrary $M_{x}\in\mathbb{M}_{x}(\mathcal{E})$ and $M_{i}=\mathbb{O}$ for all $i\in\{?\}\cup\mathbb{N}_{n}\setminus\{x\}$. 
In this case, $\{M_{?}\}\cup\{M_{i}\}_{i\in\mathbb{N}_{n}}$ clearly satisfies Inclusion~\eqref{eq:domh}; therefore, Inequality~\eqref{eq:inrt} becomes
\begin{equation}\label{eq:inrf}
\Tr\Gamma-\Tr[M_{x}(\Gamma-\gamma\eta_{x}\rho_{x})]\leqslant\gamma p_{\sf G}(\mathcal{E})
\end{equation}
by taking the limit of $p$ to $p_{\sf G}(\mathcal{E})$.

Suppose $\Gamma-\gamma\eta_{x}\rho_{x}\notin\mathbb{M}_{x}^{*}(\mathcal{E})$; then there exists $M\in\mathbb{M}_{x}(\mathcal{E})$ with $\Tr[M(\Gamma-\gamma\eta_{x}\rho_{x})]<0$.
We note that $M\in\mathbb{M}_{x}(\mathcal{E})$ implies $tM\in\mathbb{M}_{x}(\mathcal{E})$.
Thus $\{M_{?}\}\cup\{M_{i}\}_{i\in\mathbb{N}_{n}}$ consisting of $M_{x}=tM$ for $t>0$ and $M_{i}=\mathbb{O}$ for all $i\in\{?\}\cup\mathbb{N}_{n}\setminus\{x\}$ also satisfies Inclusion~\eqref{eq:domh}.
Now, Inequality~\eqref{eq:inrf} can be rewritten as
\begin{equation}\label{eq:iddf}
\Tr\Gamma-t\Tr[M(\Gamma-\gamma\eta_{x}\rho_{x})]\leqslant\gamma p_{\sf G}(\mathcal{E}).
\end{equation}
Since Inequality~\eqref{eq:iddf} is true for arbitrarily large $t>0$, $\gamma p_{\sf G}(\mathcal{E})$ can also be arbitrarily large.
However, this contradicts that both $\gamma$ and $p_{\sf G}(\mathcal{E})$ are finite.
Thus $\Gamma-\gamma\eta_{x}\rho_{x}\in\mathbb{M}_{x}^{*}(\mathcal{E})$, which completes the proof of \eqref{eq:lcob}.\qedhere
\end{proof}

\appsection{Proof of Lemma~\ref{lem:wxyz}}\label{app:wxyz}
For the case when $L=1$, Eq.~\eqref{eq:xmwy} is trivially satisfied for all $a\in\{0,1\}$.
For the case when $L>1$, we use the mathematical induction on $L$. 
We first suppose that the lemma is true for any positive integer less than or equal to $L-1$, that is,
\begin{equation}\label{eq:xmym}
\bigotimes_{l=1}^{L-1}X_{l}+(-1)^{a}\bigotimes_{l=1}^{L-1}Y_{l}=\frac{1}{2^{L-2}}\sum_{\substack{\vec{b}\in\mathbb{Z}_{2}^{L-1}\\ \omega_{2}(\vec{b})=a}}\bigotimes_{l=1}^{L-1}\big[X_{l}+(-1)^{b_{l}}Y_{l}\big]
\end{equation}
for each $a\in\{0,1\}$, and show the validity of Eq.~\eqref{eq:xmwy} for each $a\in\{0,1\}$.

For the case of $L$ copies, we have
\begin{align}\label{eq:xypv}
X+(-1)^{a}Y=&\frac{1}{2}(X_{1}\otimes\cdots\otimes X_{L-1}+Y_{1}\otimes\cdots\otimes Y_{L-1})\otimes[X_{L}+(-1)^{a}Y_{L}]\nonumber\\
&+\frac{1}{2}(X_{1}\otimes\cdots\otimes X_{L-1}-Y_{1}\otimes\cdots\otimes Y_{L-1})\otimes[X_{L}-(-1)^{a}Y_{L}]\nonumber\\
=&\frac{1}{2^{L-1}}\sum_{\substack{\vec{b}'\in\mathbb{Z}_{2}^{L-1}\\ \omega_{2}(\vec{b}')=0}}[X_{1}+(-1)^{b_{1}'}Y_{1}]\otimes\cdots\otimes[X_{L-1}+(-1)^{b_{L-1}'}Y_{L-1}]\otimes[X_{L}+(-1)^{a}Y_{L}]
\nonumber\\
&+\frac{1}{2^{L-1}}\sum_{\substack{\vec{b}'\in\mathbb{Z}_{2}^{L-1}\\ \omega_{2}(\vec{b}')=1}}[X_{1}+(-1)^{b_{1}'}Y_{1}]\otimes\cdots\otimes[X_{L-1}+(-1)^{b_{L-1}'}Y_{L-1}]\otimes[X_{L}-(-1)^{a}Y_{L}]
\nonumber\\
=&\frac{1}{2^{L-1}}\sum_{\substack{\vec{b}\in\mathbb{Z}_{2}^{L}\\ \omega_{2}(\vec{b})=a}}\bigotimes_{l=1}^{L}\big[X_{l}+(-1)^{b_{l}}Y_{l}\big]
\end{align}
for each $a\in\{0,1\}$, where $b_{l}'$ and $b_{l}$ are the $l$th entry of $\vec{b}'=(b_{1}',\ldots,b_{L-1}')$ and $\vec{b}=(b_{1},\ldots,b_{L})$, respectively.
In Eq.~\eqref{eq:xypv}, the first equality is by $X=\bigotimes_{l=1}^{L}X_{l}$ and $Y=\bigotimes_{l=1}^{L}Y_{l}$, the second equality is from the induction hypothesis in Eq.~\eqref{eq:xmym}, and the last equality holds because $-(-1)^{a}=(-1)^{1-a}$ and
\begin{align}\label{eq:ztle}
\{\vec{b}\in\mathbb{Z}_{2}^{L}\,|\,\omega_{2}(\vec{b})=a\}=&\{(b_{1}',\ldots,b_{L-1}',a)\,|\,\vec{b}'=(b_{1}',\ldots,b_{L-1}')\in\mathbb{Z}_{2}^{L-1}~\mbox{with}~\omega_{2}(\vec{b}')=0\}
\nonumber\\
&\cup\{(b_{1}',\ldots,b_{L-1}',1-a)\,|\,\vec{b}'=(b_{1}',\ldots,b_{L-1}')\in\mathbb{Z}_{2}^{L-1}~\mbox{with}~\omega_{2}(\vec{b}')=1\}.
\end{align}
Thus our lemma is true.

\appsection{Proof of Lemma~\ref{lem:ppee}}\label{app:ppee}
In this section, we use $\Pi(\mathcal{E})$ and $\Pi^{\bot}(\mathcal{E})$ to denote the projection operators onto the support and kernel of $\rho_{0}$, respectively.
For each $l=1,\ldots,L$, we also denote by $\Pi(\mathcal{E}^{l})$ and $\Pi^{\bot}(\mathcal{E}^{l})$ the projection operators onto the support and kernel of $\rho_{0}^{l}$, respectively.
Moreover, we denote by $\Pi_{x_{l}}(\mathcal{E}^{l})$ and $\Pi_{x_{l}}^{\bot}(\mathcal{E}^{l})$ the projection operators onto the support and kernel of $\mathcal{C}_{x_{l}}(\mathcal{E}^{l})\rho_{0}^{l}-\eta_{x_{l}}^{l}\rho_{x_{l}}^{l}$, respectively.

Let us consider the following equalities:
\begin{subequations}\label{eq:lsfs}
\begin{gather}
\bigotimes_{l=1}^{L}\big[\Pi(\mathcal{E}^{l})-\Pi_{x_{l}}(\mathcal{E}^{l})\big]\Pi_{x_{l}}^{\bot}(\mathcal{E}^{l})=\bigotimes_{l=1}^{L}\big[\Pi(\mathcal{E}^{l})-\Pi_{x_{l}}(\mathcal{E}^{l})\big],
\label{eq:lsfc}\\
\Pi(\mathcal{E})E\Pi(\mathcal{E})\bigotimes_{l=1}^{L}\big[\Pi(\mathcal{E}^{l})-\Pi_{x_{l}}(\mathcal{E}^{l})\big]
=\Pi(\mathcal{E})E\Pi(\mathcal{E}).
\label{eq:lssc}
\end{gather}
\end{subequations}
We note that the equalities in Eq.~\eqref{eq:lsfs} imply
\begin{eqnarray}\label{eq:fawh}
\Pi(\mathcal{E})E\Pi(\mathcal{E})\bigotimes_{l=1}^{L}\Pi_{x_{l}}^{\bot}(\mathcal{E}^{l})
&=&\Pi(\mathcal{E})E\Pi(\mathcal{E})\bigotimes_{l=1}^{L}\big[\Pi(\mathcal{E}^{l})-\Pi_{x_{l}}(\mathcal{E}^{l})\big]\Pi_{x_{l}}^{\bot}(\mathcal{E}^{l})\nonumber\\
&=&\Pi(\mathcal{E})E\Pi(\mathcal{E})\bigotimes_{l=1}^{L}\big[\Pi(\mathcal{E}^{l})-\Pi_{x_{l}}(\mathcal{E}^{l})\big]\nonumber\\
&=&\Pi(\mathcal{E})E\Pi(\mathcal{E}),
\end{eqnarray}
where the first and last equalities are by Eq.~\eqref{eq:lssc} and the second equality is by Eq.~\eqref{eq:lsfc}.
Now, we have
\begin{equation}\label{eq:twhr}
\Bigg[\bigotimes_{l=1}^{L}\Pi_{x_{l}}^{\bot}(\mathcal{E}^{l})\Bigg]
\Pi(\mathcal{E})E\Pi(\mathcal{E})
\Bigg[\bigotimes_{l=1}^{L}\Pi_{x_{l}}^{\bot}(\mathcal{E}^{l})\Bigg]
=\Bigg[\bigotimes_{l=1}^{L}\Pi_{x_{l}}^{\bot}(\mathcal{E}^{l})\Bigg]
\Pi(\mathcal{E})E\Pi(\mathcal{E})
=\Pi(\mathcal{E})E\Pi(\mathcal{E}),
\end{equation}
which is the claim of our lemma.
Thus it is enough to show the validity of Eqs.~\eqref{eq:lsfc} and \eqref{eq:lssc}.

\begin{proof}[Proof of Eq.~\eqref{eq:lsfc}]
For each $l=1,\ldots,L$, we have
\begin{eqnarray}\label{eq:pffc}
\big[\Pi(\mathcal{E}^{l})-\Pi_{x_{l}}(\mathcal{E}^{l})\big]\Pi_{x_{l}}^{\bot}(\mathcal{E}^{l})
&=&\Pi(\mathcal{E}^{l})\Pi_{x_{l}}^{\bot}(\mathcal{E}^{l})\nonumber\\
&=&\Pi(\mathcal{E}^{l})\big[\Pi(\mathcal{E}^{l})+\Pi^{\bot}(\mathcal{E}^{l})-\Pi_{x_{l}}(\mathcal{E}^{l})\big]\nonumber\\
&=&\Pi(\mathcal{E}^{l})-\Pi(\mathcal{E}^{l})\Pi_{x_{l}}(\mathcal{E}^{l}),
\end{eqnarray}
where the first equality is from the definitions of $\Pi_{x_{l}}(\mathcal{E}^{l})$ and $\Pi_{x_{l}}^{\bot}(\mathcal{E}^{l})$ and the last equality is from the definitions of $\Pi(\mathcal{E}^{l})$ and $\Pi^{\bot}(\mathcal{E}^{l})$.
The second equality in Eq.~\eqref{eq:pffc} is satisfied because 
\begin{equation}\label{eq:pieq}
\Pi(\mathcal{E}^{l})+\Pi^{\bot}(\mathcal{E}^{l})
=\Pi_{x_{l}}(\mathcal{E}^{l})+\Pi_{x_{l}}^{\bot}(\mathcal{E}^{l}).
\end{equation}

Since the support of $\mathcal{C}_{x_{l}}(\mathcal{E}^{l})\rho_{0}^{l}-\eta_{x_{l}}^{l}\rho_{x_{l}}^{l}$ is contained in that of $\rho_{0}^{l}$, the projection operator $\Pi(\mathcal{E}^{l})$ can be decomposed into the projection operator $\Pi_{x_{l}}(\mathcal{E}^{l})$ and another projection operator orthogonal to it.
In other words, $\Pi(\mathcal{E}^{l})-\Pi_{x_{l}}(\mathcal{E}^{l})$ becomes a projection operator orthogonal to $\Pi_{x_{l}}(\mathcal{E}^{l})$, that is,
\begin{equation}\label{eq:otpr}
\big[\Pi(\mathcal{E}^{l})-\Pi_{x_{l}}(\mathcal{E}^{l})\big]\Pi_{x_{l}}(\mathcal{E}^{l})=\mathbb{O}.
\end{equation}
Equation~\eqref{eq:otpr} can be equivalently written as
\begin{equation}\label{eq:pcws}
\Pi(\mathcal{E}^{l})\Pi_{x_{l}}(\mathcal{E}^{l})=\Pi_{x_{l}}(\mathcal{E}^{l}).
\end{equation}
From Eqs.~\eqref{eq:pffc} and \eqref{eq:pcws}, we have
\begin{equation}\label{eq:elpm}
\big[\Pi(\mathcal{E}^{l})-\Pi_{x_{l}}(\mathcal{E}^{l})\big]\Pi_{x_{l}}^{\bot}(\mathcal{E}^{l})
=\Pi(\mathcal{E}^{l})-\Pi_{x_{l}}(\mathcal{E}^{l}).
\end{equation}
Because Eq.~\eqref{eq:elpm} is true for each $l=1,\ldots,L$, this leads us to Eq.~\eqref{eq:lsfc}.
\qedhere
\end{proof}

\begin{proof}[Proof of Eq.~\eqref{eq:lssc}]
Suppose
\begin{equation}\label{eq:oluo}
\Pi(\mathcal{E})E\Pi(\mathcal{E})\bigotimes_{l=1}^{L}\big[
(1-b_{l})\Pi(\mathcal{E}^{l})+b_{l}\Pi_{x_{l}}(\mathcal{E}^{l})\big]=\mathbb{O}
\end{equation}
for all $\vec{b}=(b_{1},\ldots,b_{L})\in\mathbb{Z}_{2}^{L}$ with $\omega_{2}(\vec{b})=1$. 
Equation~\eqref{eq:oluo} implies
\begin{equation}\label{eq:obpi}
\Pi(\mathcal{E})E\Pi(\mathcal{E})
\big[\Pi(\mathcal{E}^{1})\otimes\cdots\otimes\Pi(\mathcal{E}^{l-1})\otimes\Pi_{x_{l}}(\mathcal{E}^{l})\otimes\Pi(\mathcal{E}^{l+1})\otimes\cdots\otimes\Pi(\mathcal{E}^{L})\big]
=\mathbb{O}
\end{equation}
for each $l=1,\ldots,L$ because Eq.~\eqref{eq:obpi} is the case of Eq.~\eqref{eq:oluo} when the $l$th entry of 
$\vec{b}\in\mathbb{Z}_{2}^{L}$ is $1$ and the remaining entries are $0$.
Thus we have
\begin{align}\label{eq:lpsw}
\Pi(\mathcal{E})E\Pi(\mathcal{E})\bigotimes_{l=1}^{L}\big[\Pi(\mathcal{E})-\Pi_{x_{l}}(\mathcal{E}^{l})\big]
=\Pi(\mathcal{E})E\Pi(\mathcal{E})&
\Big[\big(\Pi(\mathcal{E})-\Pi_{x_{1}}(\mathcal{E}^{1})\big)\otimes\Pi(\mathcal{E}^{2})\otimes\cdots\otimes\Pi(\mathcal{E}^{L})\Big]\times\cdots\nonumber\\[-1ex]
&\times\Big[\Pi(\mathcal{E}^{1})\otimes\cdots\otimes\Pi(\mathcal{E}^{L-1})\otimes\big(\Pi(\mathcal{E})-\Pi_{x_{L}}(\mathcal{E}^{L})\big)\Big]\nonumber\\
=\Pi(\mathcal{E})E\Pi(\mathcal{E})&,
\end{align}
where the first equality follows from Eq.~\eqref{eq:pcws} and the second equality uses Eq.~\eqref{eq:obpi} together with $\Pi(\mathcal{E})=\bigotimes_{l=1}^{L}\Pi(\mathcal{E}^{l})$, which follows from Eq.~\eqref{eq:aqsd}.
Equation~\eqref{eq:lpsw} completes the proof of Eq.~\eqref{eq:lssc}.

Now, we show the validity of Eq.~\eqref{eq:oluo}.
Since every element of $\mathbb{M}_{\vec{c}}(\mathcal{E})$ in Eq.~\eqref{eq:qmxe} is positive semidefinite, it follows from $E\in\mathbb{M}_{\vec{c}}(\mathcal{E})$ and the definition of $\Pi(\mathcal{E})$ that $\Pi(\mathcal{E})E\Pi(\mathcal{E})$ is also positive semidefinite.
For any $(b_{1},\ldots,b_{L})\in\mathbb{Z}_{2}^{L}$, the operator
\begin{equation}\label{eq:psto}
\bigotimes_{l=1}^{L}\big(\mathcal{C}_{x_{l}}(\mathcal{E}^{l})\rho_{0}^{l}+(-1)^{b_{l}}\eta_{x_{l}}^{l}\rho_{x_{l}}^{l}\big)
\end{equation}
is also positive semidefinite because Eq.~\eqref{eq:cxed} implies that $\mathcal{C}_{x_{l}}(\mathcal{E}^{l})\rho_{0}^{l}-\eta_{x_{l}}^{l}\rho_{x_{l}}^{l}$ is positive semidefinite for each $l=1,\ldots,L$.

For each $\vec{b}=(b_{1},\ldots,b_{L})\in\mathbb{Z}_{2}^{L}$ with $\omega_{2}(\vec{b})=1$, the positive-semidefinite operator $\Pi(\mathcal{E})E\Pi(\mathcal{E})$ is orthogonal to that in Eq.~\eqref{eq:psto} because 
\begin{eqnarray}\label{eq:crrr}
0&=&\Tr\big[E\big(\mathcal{C}_{\vec{x}}(\mathcal{E})\rho_{0}-\eta_{\vec{x}}\rho_{\vec{x}}\big)\big]\nonumber\\
&=&\Tr\big[E\Pi(\mathcal{E})\big(\mathcal{C}_{\vec{x}}(\mathcal{E})\rho_{0}-\eta_{\vec{x}}\rho_{\vec{x}}\big)\Pi(\mathcal{E})\big]\nonumber\\
&=&\Tr\big[\Pi(\mathcal{E})E\Pi(\mathcal{E})\big(\mathcal{C}_{\vec{x}}(\mathcal{E})\rho_{0}-\eta_{\vec{x}}\rho_{\vec{x}}\big)\big]\nonumber\\
&\geqslant&\frac{1}{2^{L-1}}\Tr\Big[\Pi(\mathcal{E})E\Pi(\mathcal{E})
\bigotimes_{l=1}^{L}\big(\mathcal{C}_{x_{l}}(\mathcal{E}^{l})\rho_{0}^{l}+(-1)^{b_{l}}\eta_{x_{l}}^{l}\rho_{x_{l}}^{l}\big)\Big]\nonumber\\
&\geqslant&0,
\end{eqnarray}
where the first equality uses $E\in\mathbb{M}_{\vec{c}}(\mathcal{E})$, the second equality follows from the definition of $\Pi(\mathcal{E})$, and the last equality uses the cyclicity of trace.
The inequalities in \eqref{eq:crrr} hold because $\mathcal{C}_{\vec{x}}(\mathcal{E})\rho_{0}-\eta_{\vec{x}}\rho_{\vec{x}}$ admits the representation in Eq.~\eqref{eq:crps} and $\Tr(XY)\geqslant0$ for any positive-semidefinite operators $X$ and $Y$.
Since two positive-semidefinite operators are orthogonal to each other if and only if their product yields to the zero operator, we have
\begin{equation}\label{eq:ogpt}
\Pi(\mathcal{E})E\Pi(\mathcal{E})\bigotimes_{l=1}^{L}\big(\mathcal{C}_{x_{l}}(\mathcal{E}^{l})\rho_{0}^{l}+(-1)^{b_{l}}\eta_{x_{l}}^{l}\rho_{x_{l}}^{l}\big)=\mathbb{O}.
\end{equation}

If $X$ and $\bigotimes_{l=1}^{L}Y_{l}$ are orthogonal positive-semidefinite operators, then $X$ is orthogonal to the projection operator onto the support of $\bigotimes_{l=1}^{L}Y_{l}$. Since the support of $\bigotimes_{l=1}^{L}Y_{l}$ is the tensor product of the supports of $Y_{l}$ for $l=1,\ldots,L$, it follows that $X$ is orthogonal to the tensor product of the projection operators onto the supports of $Y_{l}$ for $l=1,\ldots,L$.
For each $l=1,\ldots,L$, the projection operators onto the supports of $\mathcal{C}_{x_{l}}(\mathcal{E}^{l})\rho_{0}^{l}+\eta_{x_{l}}^{l}\rho_{x_{l}}^{l}$ and $\mathcal{C}_{x_{l}}(\mathcal{E}^{l})\rho_{0}^{l}-\eta_{x_{l}}^{l}\rho_{x_{l}}^{l}$ are $\Pi(\mathcal{E}^{l})$ and $\Pi_{x_{l}}(\mathcal{E}^{l})$, respectively, because $\mathcal{C}_{x_{l}}(\mathcal{E}^{l})\rho_{0}^{l}+\eta_{x_{l}}^{l}\rho_{x_{l}}^{l}$ has the same support as $\rho_{0}^{l}$.
Thus Eq.~\eqref{eq:ogpt} leads us to Eq.~\eqref{eq:oluo}.\qedhere
\end{proof}

 

\begin{thebibliography}{99}
\bibitem{supp} 
The support of a Hermitian operator is the subspace spanned by all the eigenvectors with nonzero eigenvalue.
\bibitem{sht}   
When $A$ and $B$ are disjoint convex sets in a real vector space $V$ with an inner product $\langle\cdot,\cdot\rangle$, there exist $x\in\mathbb{R}$ and $\vec{v}\in V$ such that $\vec{v}\neq\vec{0}$ and $\langle \vec{a},\vec{v}\rangle\leqslant x\leqslant\langle \vec{b},\vec{v}\rangle$  for all $\vec{a}\in A$ and all $\vec{b}\in B$.
\end{thebibliography}
\end{document}